\documentclass[a4paper,
					nameinlink]{article}

\usepackage{amsthm}

\usepackage[utf8]{inputenc}
\usepackage[affil-it]{authblk}
\usepackage{academicons}
\usepackage{amsmath}
\usepackage{amssymb}
\usepackage{paralist}
\usepackage{enumerate}
\usepackage{dsfont}
\usepackage{thmtools}
\usepackage{mathtools}
\usepackage{graphicx}
\usepackage{multirow}
\usepackage{xcolor}
\usepackage{tikz}
\usetikzlibrary {arrows.meta} 
\usetikzlibrary{graphs}
\usepackage{todonotes}
\usepackage{placeins}
\usepackage{nicefrac}
\usepackage{tabularx}
\usepackage{algorithm2e}
\usepackage[T1]{fontenc}
\usepackage{etoolbox}
\usepackage{hyperref}
\hypersetup{
	colorlinks,
	linkcolor={red!50!black},
	citecolor={blue!90!black!80},
	urlcolor={blue!80!black}
}
\usepackage{cleveref}

\usepackage{lineno}

{\upshape\itshape}{\upshape\rmfamily}

\newcommand{\kommentar}[1]{}
\newcommand{\Oh}{\ensuremath{\mathcal{O}}}

\newcommand{\proofpara}[1]{\medskip\noindent\textit{#1.}}

\DeclareMathOperator{\off}{off}
\DeclareMathOperator{\desc}{desc}
\DeclareMathOperator{\anc}{anc}
\DeclareMathOperator{\DP}{DP}
\DeclareMathOperator{\var}{var}
\newcommand{\Dbar}{{\ensuremath{\overline{D}}}\xspace}

\newcommand{\w}{{\ensuremath{\omega}}\xspace}

\newcommand{\yes}{{\normalfont\texttt{yes}}\xspace}
\newcommand{\no}{{\normalfont\texttt{no}}\xspace}

\newcommand{\Wh}[1]{{\normalfont\texttt{{W[#1]}}}\xspace}
\newcommand{\NP}{{\normalfont\texttt{NP}}\xspace}
\newcommand{\FPT}{{\normalfont\texttt{FPT}}\xspace}
\newcommand{\SETH}{{\normalfont\texttt{SETH}}\xspace}
\newcommand{\ETH}{{\normalfont\texttt{ETH}}\xspace}
\newcommand{\XP}{{\normalfont\texttt{XP}}\xspace}

\newcommand{\Instance}{{\ensuremath{\mathcal{I}}}\xspace}

\newcommand{\Tree}{{\ensuremath{\mathcal{T}}}\xspace}

\newcommand{\PDgiveTree}[1]{{\ensuremath{PD_{#1}}}\xspace}
\newcommand{\PD}{\PDgiveTree{\Tree}}

\newcommand{\problemdef}[3]{
	\smallskip
	
	\noindent \normalsize\textsc{#1} \smallskip \\ \nopagebreak[4]
	\begin{tabularx}{\textwidth}{@{}l@{\hspace{3pt}}X}
		\normalsize\textbf{Input:}    & \normalsize#2 \\
		\normalsize\textbf{Question:} & \normalsize#3
	\end{tabularx}
	\smallskip
}

\newcommand{\PROB}[1]{{{\normalfont\textsc{#1}}}\xspace}

\newcommand{\KP}{\PROB{Knapsack}}
\newcommand{\SubSum}{\PROB{Subset Sum}}

\newcommand{\MPDlong}{\PROB{Maximize Phylogenetic Diversity}}
\newcommand{\MPD}{\PROB{MPD}}

\usepackage{paralist}
\usepackage[a4paper, total={4.82in, 8in}]{geometry}

\newcommand{\Hbar}[1]{{\ensuremath{\overline{H_{#1}}}}\xspace}
\newcommand{\Pbar}{{\ensuremath{\overline{P}}}\xspace}
\newcommand{\Cbar}{{\ensuremath{\overline{C}}}\xspace}

\newcommand{\tPDwslong}{\PROB{Strict Time Sensitive Maximization of Phylogenetic Diversity}}
\newcommand{\tPDslong}{\PROB{Time Sensitive Maximization of Phylogenetic Diversity}}
\newcommand{\tPDws}{\PROB{\mbox{s-Time-PD}}}
\newcommand{\tPDs}{\PROB{\mbox{Time-PD}}}

\newcommand{\ctPDws}{\PROB{\mbox{c-\tPDws}}}
\newcommand{\ctPDs}{\PROB{\mbox{c-\tPDs}}}
\newcommand{\cBartPDslong}[1]{\PROB{extinction $#1$-colored \tPDslong}}
\newcommand{\cBartPDs}[1]{\PROB{\mbox{ex-$#1$-c-\tPDs}}}

\DeclareMathOperator{\ex}{ex}
\newcommand{\Pvx}{\ensuremath{P_{v,x}}\xspace}
\newcommand{\mcA}{\ensuremath{\mathcal A}\xspace}

\newtheorem{lemma}{Lemma}
\newtheorem{theorem}{Theorem}
\newtheorem{proposition}{Proposition}
\newtheorem{definition}{Definition}
\newenvironment{proofm}[1]
  {\noindent \textbf{Proof ({\normalfont #1}).}}
  {\hfill$\square$ \\ \goodbreak}
\renewenvironment{proof}
  {\noindent \textbf{Proof.}}
  {\hfill$\square$ \\ \goodbreak}

\newcommand\proofpart[1]{

\smallskip #1\\ \noindent}

\newcommand\mails{Electronic Adresses: \texttt{m.e.l.jones(at)tudelft.nl} and \texttt{j.t.schestag(at)uni-jena.de}}

\newcommand\orcidID[1]{[#1]}
\definecolor{orcidlogocol}{HTML}{A6CE39}

\begin{document}

\title{Maximizing Phylogenetic Diversity under Time Pressure: Planning with Extinctions Ahead}

\author[1]{Mark Jones
	\footnote{Partially supported by Netherlands Organisation for Scientific Research~(NWO) grant OCENW.KLEIN.125.
		Orcid: \href{https://orcid.org/0000-0002-4091-7089}{\orcidID{0000-0002-4091-7089}}}${}^,$}

\author[1,2]{Jannik Schestag
	\footnote{The research was partially carried out during an extended research visit of Jannik Schestag at TU Delft, The Netherlands. We thank the German Academic Exchange Service (DAAD), project no. 57556279, for the funding.
		Orcid: \href{https://orcid.org/0000-0001-7767-2970}{\orcidID{0000-0001-7767-2970}}\\
	\mails}${}^,$}

\affil[1]{TU Delft, Faculteit Elektrotechniek, Wiskunde en Informatica, Mekelweg 4, 2628 CD Delft, The Netherlands}

\affil[2]{Friedrich-Schiller-Universität Jena, Fachbereich Mathematik und Informatik, Ernst-Abbe-Platz 2, 07743 Jena, Germany}

\date{\today}

\maketitle

\begin{abstract}
	Phylogenetic Diversity (PD) is a measure of the overall biodiversity of a set of present-day species~(taxa) within a phylogenetic tree.
	In \MPDlong (\MPD) one is asked to find a set of taxa (of bounded size/cost) for which this measure is maximized.
	\MPD is a relevant problem in conservation planning, where there are not enough resources to preserve all taxa and minimizing the overall loss of biodiversity is critical.
	We consider an extension of this problem, motivated by real-world concerns, in which each taxon  not only requires a certain amount of time to save, but also has an extinction time after which it can no longer be saved. 
	In addition there may be multiple teams available to work on preservation efforts in parallel;
	we consider two variants of the problem based on whether teams are allowed to collaborate on the same taxa.
	These problems have much in common with machine scheduling problems, (with taxa corresponding to tasks and teams corresponding to machines), but with the objective function (the phylogenetic diversity) inspired by biological considerations.
	Our extensions are, in contrast to the original \MPD, \NP-hard, even in very restricted cases.
	We provide several algorithms and hardness-results and thereby show that the problems are fixed-parameter tractable (\FPT) when parameterized the target phylogenetic diversity, and that the problem where teams are allowed to collaborate is \FPT when parameterized the acceptable loss of diversity.
\end{abstract}

\newpage
\section{Introduction}
The 2023 report of the Intergovernmental Panel on Climate Change (IPCC)~\cite[Page~73]{ipcc2023} warns that as the average global temperature increases, thousands of plant and animal species around the world face increasing risk of extinction.
With a ``climate breakdown''~\cite{climatebreakdown} on the horizon, the remaining time for intervention is rapidly running out. 
In order to maximize the effects of conservation efforts and preserve as much biodiversity as possible, efficient schedules for operating teams are needed.

In this paper we consider an extension of the classic \MPDlong (\MPD) problem, in which species~(taxa) have differing \emph{extinction times}, after which they will die out if they have not been saved.
Phylogenetic diversity, first introduced in 1992 by Faith~\cite{FAITH1992} is a measure of the amount of biodiversity within a set of taxa. 
Given a (rooted) phylogenetic tree whose leaves are labeled with elements from a set of taxa $X$, and whose edges are weighted with an integer value representing the evolutionary distance between two taxa, the \emph{phylogenetic diversity} of any subset $A \subseteq X$ is the total weight of all edges on a path from the root to any leaf in $A$.
Intuitively, the phylogenetic diversity measure captures the expected range of biological features of a given set of taxa.
This measure forms the basis of the Fair Proportion Index and the Shapley Value~\cite{haake2008shapley,Hartmann2013TheEO,redding2006incorporating}, which are used to evaluate the individual contribution of individual taxa to overall biodiversity. 

In~\MPD, we are given a phylogenetic tree $\Tree$ and integers $k$, and~$D$ as input, and we are asked if there exists a subset of $k$ taxa whose phylogenetic diversity is at least $D$.
\MPD is polynomial-time solvable by a greedy algorithm~\cite{steel,Pardi2005}.
However if each taxon has an associated integer cost,
the problem becomes \NP-hard~\cite{ResourceAwarePD}.

In this extension of \MPD, the cost of a taxon may, just as well as financial or space capacities, be used to represent the fact that different taxa may take different amounts of time to save from extinction.
However, to the best of our knowledge, previously studied versions of \MPD do not take into account that different taxa may have different amounts of \emph{remaining} time before extinction.
Thus, to ensure that a set of taxa can be saved with the available resources, it is not enough to guarantee that their total cost is below a certain threshold.
One also needs to ensure that there is a schedule under which each taxon is saved before its moment of extinction.

In this paper we take a first step to addressing this issue. 
We introduce two extensions of~\MPD, denoted \tPDslong(\tPDs) and \tPDwslong (\tPDws), in which each taxon has an associated \emph{rescue length} (the amount of time it takes to save them) and also an \emph{extinction time} (the time after which the taxon can not be saved anymore). In each problem there is a set of available teams who can work towards saving the taxa; under \tPDs different teams may collaborate on saving a taxon while in \tPDws they may not. (See~\Cref{sec:prelim} for formal definitions).

These problems have much in common with machine scheduling problems, insofar as we may think of the taxa as corresponding to jobs with a certain due date and the teams as corresponding to machines. One may think of \tPDs / \tPDws as machine scheduling problems, in which the objective to be maximized is the phylogenetic diversity (as determined by the input tree $\Tree$) of the set of completed tasks (which are the saved taxa).

Both problems turn out to be \NP-hard; this result is perhaps unsurprising given their close relation to scheduling problems, but stands in contrast to the classic \MPD which can be solved greedily in polynomial time.
Given the \NP-hardness of our versions of \MPD, we turn to parameterized complexity and consider how the complexity of \tPDs and \tPDws vary with respect to various parameters of the input.
In particular we show that both problems are fixed-parameter tractable (\FPT) with respect to the desired phylogenetic diversity $D$, and with respect to the acceptable amount of diversity loss (\Dbar).
Both problems are also \FPT with respect to the available person-hours $H_{\max_{\ex}}$.
We note that the available person-hours are an upper bounded for the number of teams $|T|$ and the maximum extinction time $\max_{\ex}$ of any taxa. Further, $H_{\max_{\ex}}$ itself is bounded in $|T| \cdot \max_{\ex}$.
(We note that the problems remain \NP-hard even if $|T|=1$; for $\max_{\ex}$ the parameterized complexity remains open).
As well as our \FPT results, 
we prove a number of \XP and pseudo-polynomial time results.
See \Cref{tab:results} for a full list of results.

\paragraph{Structure of the paper}
In \Cref{sec:prelim} we formally define the problems \tPDs and \tPDws, give an overview of previous results and our contribution, and prove some simple intial results.
In \Cref{sec:D} we introduce an \FPT algorithm for \tPDs and \tPDws parameterized by the target  $D$, and in \Cref{sec:Dbar} we give an \FPT algorithm for \tPDs parameterized by the acceptable loss of diversity $\Dbar$.
In \Cref{sec:other} we prove a number of other parameterized algorithms.
Finally, in \Cref{sec:discussion} we discuss open problems and future research.

\section{Preliminaries}
\label{sec:prelim}

\subsection{Definitions}
For a positive integer $a$
we denote by $[a]$ the set $\{1,2,\dots,a\}$, and
by $[a]_0$ the set $\{0\}\cup [a]$.
We generalize functions $f:A\to B$ to handle subsets $A'\subseteq A$ by $f(A') = \bigcup_{a\in A'} f(a)$ or $f(A') = \sum_{a\in A'} f(a)$ depending on the context.

Throughout the paper we use $\mathcal O^*$ to describe the running times of algorithms in which we omit factors polynomial in the input size.

\paragraph*{Phylogenetic Trees and Phylogenetic Diversity.}
For a given set $X$, a \textit{phylogenetic $X$-tree~$\Tree=(V,E,\w)$} is a connected cycle-free directed graph with \textit{edge-weight} function~$\w: E\to \mathbb{N}_{>0}$ and a single vertex of no incoming edges and out-degree two (the \textit{root}, usually denoted with $\rho$), in which the \textit{leaves}, that are the vertices with no outgoing edges, have in-degree one and are bijectively labeled with elements from a set $X$, and such that all other vertices have in-degree two and out-degree at least two.
In biological applications, the set $X$ is a set of \textit{taxa} (species), the internal vertices of~$\Tree$ correspond to biological ancestors of these taxa and~$\w(e)$ describes the phylogenetic distance between the endpoints of~$e$ (as these endpoints correspond to distinct taxa, we may assume this distance is greater than zero).

We refer to the non-leaf vertices of a tree as the \emph{internal vertices}.
For a directed edge $uv \in E$, we say $u$ is the \emph{parent} of $v$ and $v$ is a \emph{child} of $u$.
If there is a directed path from $u$ to $v$ in $\Tree$ (including when $u=v$ or when $u$ is a child of $v$), we say that $u$ is an \emph{ancestor} of $v$ and $v$ is a \emph{descendant}. If in addition $u \neq v$ we say $u$ is a \emph{strict ancestor} of $v$ and $v$ a \emph{strict descendant} of $u$.
The sets of ancestors and descendants of $v$ are denoted $\anc(v)$ and $\desc(v)$, respectively.
The \emph{offspring} $\off(v)$ of a vertex $v$ is the set of all descendants of $v$ in $X$.
Generalized we denote $\off(e) = \off(v)$ for an edge $e = uv \in E$.
For two edges $e$ and $\hat e$ incident to the same vertex $v$, we say $\hat e$ is the \emph{parent-edge} of $e$ if $\hat e = uv$ and $e = vw$. If $e = vu$ and $\hat e = vw$, we say $\hat e$ is a \emph{sibling-edge} of $e$.

Given a subset of taxa $A \subseteq X$,
let $E_{\Tree}(A)$ denote the set of edges in $e \in E$ with $\off(e)\cap A \neq \emptyset$.
The \emph{phylogenetic diversity} $\PD(A)$ of $A$ is defined by 

	\begin{equation}
		\label{eqn:PDdef}
		\PD(A) := \sum_{e \in E_{\Tree}(A)}\w(e).
	\end{equation}


\paragraph*{Problem Definitions and Parameterizations.}
In the following,
we are given a set of taxa $X$ and a phylogenetic $X$-tree $\Tree$, which will be used to calculate the phylogenetic diversity of any subset of taxa $A \subseteq X$.
In addition, for each taxon $x \in X$ we are given an integer \emph{extinction time $\ex(x)$} representing the amount of time remaining to save that taxon before it goes extinct, and an integer \emph{rescue length} $\ell(x)$ representing how much time needs to be spent in order to save that taxon. 
Thus, we need to spend $\ell(x)$ units of time before $\ex(x)$ units of time have elapsed, if we wish to save $x$ from extinction.

In addition, we are given a set of teams  $T = \{t_1,\dots, t_{|T|}\}$ that are available to work on saving taxa, where each team $t_i$ is represented by a pair of integers $(s_i,e_i)$ with $0 \leq s_i < e_i$.
Here $s_i$ and $e_i$ represent the starting time and the ending time of team $t_i$, respectively.
Thus, team $t_i$ is available for $(e_i-s_i)$~units of time, starting after $s_i$ time steps.
For convenience, we refer to the units of time in this paper as \emph{person-hours}, although of course in practice the units of time may be days, weeks or months.
We define $\mathcal{H}_T : = \{(i,j) \in T\times \mathbb{N} \mid t_i \in T, s_i < j \leq e_i\}$.
That is, $\mathcal{H}_T$ is the set of all pairs $(i,j)$ where team $t_i$ is available to work in \emph{timeslot} $j$.
For $j \in [\var_{\ex}]$, we define $H_j: = |\{(i,j')\in \mathcal{H}_T \mid j' \le \ex_j\}|$. That is, $H_j$ is the number of person-hours available until time $\ex_j$. 

Then for any subset of taxa $A \subseteq X$, a \emph{$T$-schedule} is a function $f:\mathcal{H}_T\rightarrow A\cup\{\textsc{none}\}$, that is, a function where each available timeslot for each team is mapped to a taxon in $A$ (or $\textsc{none}$).
Intuitively, $f$ shows which teams will work to save which taxa in $A$, and at which times. See \Cref{fig:example-scheduling} for an example.
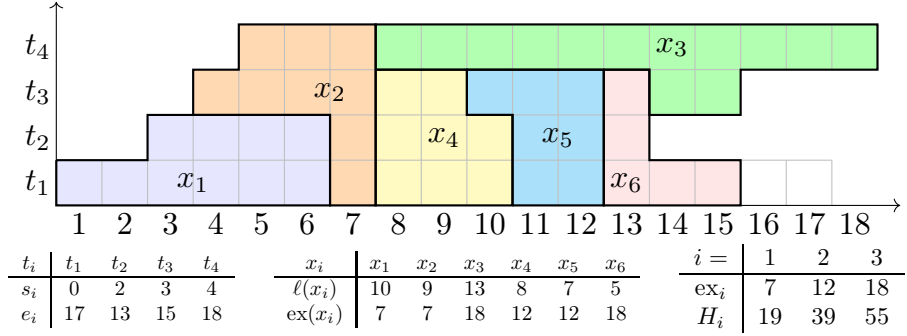
\begin{figure}[t]
	\centering
	\begin{tikzpicture}[scale=0.6,every node/.style={scale=1.2}]
		\tikzstyle{txt}=[circle,fill=none,draw=none,inner sep=0pt]
		
		\fill[blue!10] (0,0) rectangle (6,1);
		\fill[blue!10] (2,1) rectangle (6,2);
		\node[txt] at (3,0.5) {$x_1$};
		
		\fill[orange!30] (3,2) rectangle (8,3);
		\fill[orange!30] (4,3) rectangle (7,4);
		\fill[orange!30] (6,0) rectangle (8,2);
		\node[txt] at (6,2.5) {$x_2$};
		
		\fill[green!30] (7,4) rectangle (18,3);
		\fill[green!30] (13,2) rectangle (15,3);
		\node[txt] at (13.5,3.5) {$x_3$};
		
		\fill[yellow!30] (7,0) rectangle (9,3);
		\fill[yellow!30] (9,0) rectangle (10,2);
		\node[txt] at (8.5,1.5) {$x_4$};
		
		\fill[cyan!30] (9,2) rectangle (12,3);
		\fill[cyan!30] (10,0) rectangle (12,2);
		\node[txt] at (11,1.5) {$x_5$};
		
		\fill[red!10] (12,0) rectangle (15,1);
		\fill[red!10] (12,1) rectangle (13,3);
		\node[txt] at (12.5,0.5) {$x_6$};

		\draw[gray!50] (0,1) -- (17,1);
		\draw[gray!50] (2,2) -- (15,2);
		\draw[gray!50] (3,3) -- (18,3);
		\draw[gray!50] (4,4) -- (18,4);
		
		\foreach \i in {1,...,3}
		\draw[gray!50] (\i,0) -- (\i,\i);
		
		\foreach \i in {4,...,13}
		\draw[gray!50] (\i,0) -- (\i,4);
		
		\foreach \i in {14,...,17}
		\draw[gray!50] (\i,0) -- (\i,1);

		\foreach \i in {14,15}
		\draw[gray!50] (\i,2) -- (\i,3);

		\foreach \i in {14,...,18}
		\draw[gray!50] (\i,3) -- (\i,4);
		
		\foreach \i in {1,...,18}
		\node[txt] at (\i-.5,-.4) {\i};
		
		\foreach \i in {1,...,4}
		\node[txt] at (-.4,\i-.5) {$t_\i$};
		
		\draw[->] (0,0) -- (0,4.5);
		\draw[->] (0,0) -- (18.5,0);
		
		\draw[thick] (0,1) -- (2,1) -- (2,2) -- (6,2) -- (6,0) -- (0,0) -- (0,1);
		\draw[thick] (3,2) -- (3,3) -- (4,3) -- (4,4) -- (7,4) -- (7,0) -- (6,0) -- (6,2) -- (3,2);
		\draw[thick] (7,3) -- (7,4) -- (18,4) -- (18,3) -- (15,3) -- (15,2) -- (13,2) -- (13,3) -- (7,3);
		\draw[thick] (7,3) -- (7,0) -- (10,0) -- (10,2) -- (9,2) -- (9,3) -- (7,3);
		\draw[thick] (10,0) -- (10,2) -- (9,2) -- (9,3) -- (12,3) -- (12,0) -- (10,0);
		\draw[thick] (12,3) -- (12,0) -- (15,0) -- (15,1) -- (13,1) -- (13,3) -- (12,3);
	\end{tikzpicture}
	\resizebox{.25\columnwidth}{!}{%
	\begin{tabular}{c|cccc}
	$t_i$ & $t_1$ & $t_2$ & $t_3$ & $t_4$\\
	\hline
	$s_i$ & 0 & 2 & 3 & 4\\
	$e_i$ & 17 & 13 & 15 & 18
	\end{tabular}
	}
	\;\;
	\resizebox{.4\columnwidth}{!}{%
	\begin{tabular}{c|cccccc}
	$x_i$ & $x_1$ & $x_2$ & $x_3$ & $x_4$ & $x_5$ & $x_6$\\
	\hline
	$\ell(x_i)$ & 10 & 9 & 13 & 8 & 7 & 5\\
	$\ex(x_i)$ & 7 & 7 & 18 & 12 & 12 & 18
	\end{tabular}
	}
	\;\;
	\resizebox{.25\columnwidth}{!}{%
	\begin{tabular}{c|ccc}
	$i=$ & $1$ & $2$ & $3$\\
	\hline
	$\ex_i$ & 7 & 12 & 18\\
	$H_i$ & 19 & 39 & 55
	\end{tabular}
	}
	\caption{This is a hypothetical valid schedule saving the set of taxa $\{x_1,\dots,x_6\}$.
		Each cube marks a tuple $(i,j) \in \mathcal{H}_T$.
		}
	\label{fig:example-scheduling}
\end{figure}%

We say $f$ is \emph{valid} if $\ex(x) \geq j$ for each $x \in A$, and $(i,j) \in \mathcal{H}$ with $f((i,j)) = x$.
That is, $f$ does not assign a team to work on taxon $x$ after its extinction time.
We say that $f$ \emph{saves $A$} if $f$ is valid and $|f^{-1}(x)| \geq \ell(x)$ for all $x \in A$.
That is, a schedule $f$ saves $A \subseteq X$ if every taxon $x \in A$ has at least $\ell(x)$ person-hours assigned to $x$ by its extinction time $\ex(x)$.

The definition of a valid schedule allows for several teams to be assigned to the same taxon at the same time, so that for example the task of preserving a taxon can be done in half the time by using twice the number of teams. Whether this is realistic or not depends on the nature of the tasks involved in preservation. For instance the task of preparing a new enclosure for animals might be completed faster by several teams working in parallel, whilst the task of rearing infants to adulthood cannot be sped up in the same way. 
Due to these concerns, we will consider two variations of the problem, differentiating in whether a schedule must be \emph{strict}.


We say a $T$-schedule $f$ saving $A$ is \emph{strict} if $|\{ i \mid f((i,j)) = x \}| = 1$ for each $x \in A$.
That is, there is only one team $t_i$ assigned to save each taxon $x$.
We note without proof that we may assume $f^{-1}(x) = \{(i,j+1), (i,j+2), \dots, (i,j+\ell(x)\}$ for some $j \ge s_i$---that is, once started, team $t_i$ continuously works on $x$.
We note that in a non-strict schedule $f$, multiple teams may work on the same taxa $x$ at once---it is even possible that $f(i,j) = f(i',j) = x$ for $i\neq i$.

Formally, the problems we regard in this paper are as follows.
\problemdef{\tPDslong (\tPDs)}{
A directed $X$-tree $\Tree = (V,E,\w)$, integers $\ex(x)$ and $\ell(x)$ for each $x\in X$, a set of teams $T$, and a target diversity $D\in \mathbb{N}$.
}{
Is there a valid $T$-schedule saving $A$, for some $A\subseteq X$ such that $PD_\Tree(A)\ge D$?
}

The problem \tPDws is the same as \tPDs, except for the restriction that the valid $T$-schedule should be strict.

\problemdef{\tPDwslong\\ (\tPDws)}{
A directed $X$-tree $\Tree = (V,E,\w)$, integers $\ex(x)$ and $\ell(x)$ for each $x\in X$, a set of teams $T$, and a target diversity $D\in \mathbb{N}$.
}{
Is there a strict valid $T$-schedule saving $A$, for some $A\subseteq X$ such that $PD_\Tree(A)\ge D$?
}

Observe that if there is only one team every valid schedule is also strict. Thus, an instance $\Instance = (\Tree, \ell, \ex, T, D)$ with $|T|=1$ is a \yes-instance of \tPDs if and only if it is a \yes-instance of \tPDws.
\Cref{lem:scheduleCondition} and \ref{lem:strictScheduleCondition} elaborate on the conditions in these questions in more detail.

\paragraph{Parameterized Complexity.}

Throughout this paper we consider a number of parameterizations of \tPDs and \tPDws. 
For a detailed introduction to parameterized complexity refer to the standard monographs~\cite{cygan,downeybook}; we only give a brief overview here.

A parameterization of a problem $\Pi$ associates with each input instance $\Instance$ of $\Pi$ the size $k$ of a specific \emph{parameter}.
A parameterized problem $\Pi$ is \emph{fixed-parameter tractable} (\FPT), or \emph{slice-wise polynomial} (\XP), with respect to some parameter if there exists an algorithm solving every instance $(\Instance,k)$ of $\Pi$ in time $f(k)\cdot |\Instance|^{\Oh(1)}$, or in time $|\Instance|^{f(k)}$, respectively.
Here in both cases, $f$ is some computable function only depending on $k$.
Parameterized problems that are \emph{\Wh{1}-hard} are believed not to be \FPT.
%
For convenience reasons, the word parameterized is usually omitted.



\paragraph{Additional definitions.}
We now introduce some additional definitions that will be helpful for the parameterizations and proofs that follow.

Given an instance $\Instance$ of \tPDs or \tPDws with target diversity $D$, we define $\Dbar : = \PD(X)-D = \sum_{e \in E}\w(e) - D$. Thus $\Dbar$ is the acceptable loss of diversity---if we save a set of taxa $A \subseteq X$ with $\PD(A)\geq D$, then the amount of diversity we lose from $\Tree$ as a whole is at most $\Dbar$.

Let $\var_{\ex} := |\{\ex(x) \mid x \in X\}|$ and $\max_{\ex} := \max\{\ex(x) \mid x\in X\}$. That is, $\var_{\ex}$ is the number of different extinction times for taxa in $X$, and $\max_{\ex}$ is the latest extinction time. 
Let $\ex_1 < \ex_2 < \dots < \ex_{\var_{\ex}} = \max_{\ex}$ denote the elements of~$\ex(X)$.
For each $j\in[\var_{\ex}]$ the \emph{class $Y_j\subseteq X$} is the set of taxa $x$ with $\ex(x)=\ex_j$ and we define $Z_j = Y_1 \cup \dots \cup Y_j$.
Further, we define $\ex^*(x) = j$ for each $x\in Y_j$.

We define $\Hbar j$ to be $\sum_{x\in Z_j} \ell(x) - H_j$ for each $j\in [\var_{\ex}]$. That is,  $\Hbar j$ is the difference between the number of person-hours needed to save all taxa in $Z_j$, and the number of person-hours available to save them.

Along similar lines to $\var_{\ex}$ and $\max_{\ex}$, we define $\var_\ell:= |\{\ell(x) \mid x\in X\}|$ and $\max_\ell:= |\{\ell(x) \mid x\in X\}|$. That is, $\var_\ell$ is the number of different rescue lengths for taxa in $X$, and $\max_\ell$ is the largest rescue length.  We let $\ell_1 < \ell_2 < \dots < \ell_{\var_\ell} = \max_\ell$ denote the elements of $\{\ell(x) \mid x \in X\}$.


\subsection{Related Work}

\paragraph*{Phylogenetic Diversity}
In the classic \MPD, we are given a directed $X$-tree $\Tree = (V,E,\w)$ and two positive integers $k$ and $D$ as input.
Here, $k$ stands for a solution size and $D$ for a target budget.
It is asked if there exists a set $|A|\subseteq X$ of size at most $k$ such that $PD_\Tree(A)\ge D$.
\MPD is polynomial-time solvable---a greedy algorithm was developed in~\cite{steel}, and shown to always return the optimal solution in \cite{Pardi2005}.
Even in practice, instances of \MPD with thousands of taxa can be solved within a few seconds~\cite{MinhPDinSeconds}.

A natural extension is to assign each taxon an integer cost, and replace the requirement that $|A|\leq k$ with the requirement that the total cost of a set of taxa $A$ is at most some budget $B$.
This version of the problem is \NP-hard, though there exists a pseudo-polynomial time algorithm~\cite{ResourceAwarePD}.
In some models this is further extended. Here, one can even choose between several projects for a taxon that vary in cost and success-probability~\cite{billionnet13,billionnet17,GNAP}.

\MPD has also been studied in cases where all taxa associated with an area are saved together~\cite{rodrigues2002,MoultonSempleSteel,rodrigues_brooks_gaston_2005}, and in cases when one must consider the phylogenetic diversity measure from several trees~\cite{Bordewich2Trees,SNM2009}.
Many of these variations of \MPD have a polynomial time $(1 + 1/e)$ approximation algorithm~\cite{BudgetedNatureReserveSelection}.
Extensions of phylogenetic diversity to phylogenetic networks have also been considered, both for splits systems~\cite{SNM2009,PDSplitSystems,CircularSplitSystems}
and more recently for explicit phylogenetic networks~\cite{WickeFischer2018,bordewich,MAPPD}.

%
%
%


\paragraph*{Scheduling}
Scheduling problems are denoted by a three-field notation introduced by Graham~et~al.~\cite{graham1979}.
Herein, many problems are written as a triple $\alpha|\beta|\gamma$, where $\alpha$ is the machine environment, $\beta$ are job characteristics and scheduling constraints, and $\gamma$ is the objective function.
For a more detailed view, we refer to~\cite{graham1979,mnich}.

The scheduling problem most closely related to the problems studied in this paper is $Pt||\sum w_j (1-U_j)$. In this problem we are given a set of jobs, each with an integer weight, processing time, and due date. We are also given $t$ identical machines that can each process one job at a time. The task is to schedule jobs on the available machines in such a way that we maximize the total weight of jobs completed before the due date. (Here $w_j$ denotes the weight of job $j$, and $U_j = 1$ if job $j$ is not completed on time, and $U_j =0$ otherwise).
This is similar to \tPDws, in the case that the $t$ available teams have identical starting and ending times. We may think of taxa as analogous to jobs, with the extinction times corresponding to due dates.
The key difference is that in \tPDws, rather than maximizing the total weight of the completed jobs, we aim to maximize their phylogenetic diversity, as determined by the input tree $\Tree$. 


Although approximation algorithms for scheduling problems are common, parameterized algorithms for scheduling problems are rare~\cite{mnich}.
The most commonly investigated special case of \tPDws is $1||\sum w_j (1-U_j)$---that\linebreak is~$Pt||\sum w_j (1-U_j)$ with only one machine.
This problem is weakly \NP-hard and is solvable in pseudo-polynomial time~\cite{lawler1969,sahni1976},
while the unweighted version~$1||\sum (1-U_j)$ is solvable in polynomial time~\cite{moore1968,maxwell1970,sidney1973}.
Parameters studied for~$1||\sum w_j (1-U_j)$ include the number of different due dates, the number of different processing times, and the number of different weights.
The problem  is \FPT when parameterized by the sum of any two of these parameters~\cite{hermelin}.
When parameterized by one of the latter two parameters, it is \Wh{1}-hard~\cite{heeger2024} but \XP~\cite{hermelin}.
Also a version has been studied in the light of parameterized algorithms in which there are also few distinct deadlines~\cite{heeger2023}.

\subsection{Our contribution}
With \tPDs and \tPDws we introduce problems on maximizing phylogentic diversity in which extinction times of taxa may vary.

We provide a connection to the well-regarded field of scheduling and
we analyze \tPDs and \tPDws within the framework of parameterized complexity.
Our most important results are $\Oh^*(2^{2.443\cdot D + o(D)})$-time algorithms for \tPDs and \tPDws (\Cref{sec:D}),
and a $\Oh^*(2^{6.056\cdot \Dbar + o(\Dbar)})$-time algorithm for \tPDs (\Cref{sec:Dbar}).
A detailed list of known results for \tPDs and \tPDws is given in \Cref{tab:results}.
In this table and throughout the paper, we adopt the common convention that $n:= |X|$, the number of taxa.

A key challenge for the design of \FPT algorithms for \tPDs and \tPDws is the combination of the tree structure and extinction time constraints. 
Tree-based problems often suit a dynamic programming (DP) approach, where partial solutions are constructed for subtrees of the input tree, (starting with the individual leaves and proceeding to larger subtrees). Scheduling problems with due dates (such as our extinction times) may also suit a DP approach, where partial solutions are constructed for subsets of tasks with due date below some bound. 
These two DP approaches have a conflicting structure for the desired processing order, which makes designing a DP algorithm for \tPDs and \tPDws more difficult than it may first appear.
Our solution involves the careful use of color-coding to reconcile the two approaches.
We believe this approach may also be applicable to other extensions of \MPD.

\begin{table}[t]
	\centering
	\footnotesize
	\caption{Parameterized complexity results for \tPDs and \tPDws.}
	\label{tab:results}
	\resizebox{\columnwidth}{!}{%
	\begin{tabular}{lllll}
		\hline
		Parameter & \multicolumn{2}{c}{\tPDs} & \multicolumn{2}{c}{\tPDws}\\
		\hline
		Diversity $D$ & \FPT $\Oh^*(2^{2.443\cdot D + o(D)})$ & Thm.~\ref{thm:D} & \FPT $\Oh^*(2^{2.443\cdot D + o(D)})$ & Thm.~\ref{thm:D}\\
		Diversity loss $\Dbar$ \; & \FPT $\Oh^*(2^{6.056\cdot \Dbar + o(\Dbar)})$ & Thm.~\ref{thm:DBar} & \NP-hard even if $\Dbar = 0$ & \cite{garey1978}\\
		\# Taxa $n$ & \FPT $\Oh^*(2^{n})$  & Prop.~\ref{prop:X}\ref{prop:X-s} & \FPT $\Oh^*(2^{n} \cdot n!)$  & Prop.~\ref{prop:X}\ref{prop:X-ws}\\
		\# Teams $|T|$ & \NP-hard even if $|T| = 1$  & Prop.~\ref{prop:Knapsack}\ref{thm:KP-NP} \; & \NP-hard even if $|T| = 1$  & Prop.~\ref{prop:Knapsack}\ref{thm:KP-NP}\\
		Solution size $k$ & \Wh{1}-hard & Prop.~\ref{prop:Knapsack}\ref{thm:KP-W1} & \Wh{1}-hard & Prop.~\ref{prop:Knapsack}\ref{thm:KP-W1}\\
		\hline
		Max time $\max_{\ex}$ \; & \textit{open} & & \textit{open} & \\
		Unique times $\var_{\ex}$ \; & \NP-hard even if 1 & Prop.~\ref{prop:Knapsack}\ref{thm:KP-NP} & \NP-hard even if 1 & Prop.~\ref{prop:Knapsack}\ref{thm:KP-NP}\\
		Unique lengths $\var_\ell$ \; & \textit{open} & & \Wh{1}-hard & \cite{heeger2024}\\
		\hline
		Person-hours $H_{\max}$ & \FPT $\Oh^*((|T|+1)^{2\max_{\ex}})$ & Prop.~\ref{prop:T+maxr}\ref{prop:s-T^maxr} & \FPT $\Oh^*(3^{|T|+\max_{\ex}})$ & Prop.~\ref{prop:T+maxr}\ref{prop:ws-T+maxr}\\
		& \FPT $\Oh^*((H_{\max_{\ex}})^{2\var_{\ex}})$ & Prop.~\ref{prop:T+maxr}\ref{prop:s-T*maxr} & & \\
		$\var_\ell+\var_{\ex}$ & \XP $\Oh(n^{2\var_\ell\cdot\var_{\ex}+1})$ & Prop.~\ref{prop:varl+varr} & \textit{open} & \\
		\hline
	\end{tabular}
	}
	
\end{table}

\subsection{Observations}
In the following, we start with some easy observations.

\begin{lemma}
	\label{lem:scheduleCondition}
	There exists a valid $T$-schedule saving $A \subseteq X$ if and only\linebreak if $\sum_{x\in A\cap Z_j} \ell(x) \le H_j$ for each $j\in [\var_{\ex}]$.
\end{lemma}
\begin{proof}
	Recall that $\ex_j$ is the $j$th extinction time and $\ex^*(x) = j$ for each $x\in Y_j$
	
	Suppose first that $f$ is a valid schedule saving $A$.
	Then for each $x\in A$, the set $f^{-1}(x)$ contains at least $\ell(x)$ pairs $(i,j')$ with $j' \leq \ex(x)$. It follows that $f^{-1}(Z_j)$ contains at least $\sum_{x \in A\cap Z_j} \ell(x)$ pairs $(i,j')$ with $j' \leq \ex_j$, and so  $\sum_{x\in A\cap Z_j} \ell(x) \le |\{(i,j')\in \mathcal{H}_T \mid j' \le \ex_j\}| = H_j$ for each $j\in [\var_{\ex}]$.
	
	Conversely, suppose $\sum_{x\in A\cap Z_i} \ell(x) \le H_j$ for each $j\in [\var_{\ex}]$. 
	Then order the elements of $\mathcal{H}_T$ such that $(i',j')$ appears before $(i'',j'')$ if $j' < j''$, and order the elements of $A$ such that $x$ appears before $y$ if $\ex^*(x) < \ex^*(y)$ (thus, all elements of $A\cap Y_j$ appear before all elements of $A\cap Y_{j+1}$).
	Now define $f:\mathcal{H}_T \rightarrow A\cup\{\textsc{none}\}$ by repeatedly choosing the first available taxon $x$ of $A$, and assigning it to the first available $\ell(x)$ elements of $\mathcal{H}_T$.
	Then the first $\sum_{x\in A\cap Z_j} \ell(x) \le H_j$ elements of $\mathcal{H}_T$ are used to save taxa in $Z_j$, and these elements of $\mathcal{H}_T$ are all of the form $(i,j')$ for $j' \leq \ex_j$.
	It follows that $f$ is a valid schedule saving $A$.
\end{proof}

By \Cref{lem:scheduleCondition}, we have that $\Instance$ is a \yes-instance of \tPDs (or \tPDws with $|T|=1$) if and only if there exists a set $S\subseteq X$ with $\PD(S)\geq D$ such that $\sum_{x\in A\cap Z_i} \ell(x) \le H_i$ for each $i\in [\var_{\ex}]$.
We will refer to such a set $S$ as a \emph{solution} for $\Instance$.
The following lemma follows from the definitions of valid and strict schedules.

\begin{lemma}
	\label{lem:strictScheduleCondition}
	There exists a strict valid $T$-schedule saving $A \subseteq X$ if and only if there is a partition of $A$ into sets $A_1\dots, A_{|T|}$ such that for each $i \in [|T|]$, there is a valid $\{t_i\}$-schedule saving $A_i$. 
\end{lemma}

For the paramterized side of view,
we first show that \tPDs and \tPDws are \NP-hard, even for quite restricted instances. 
Our proof adapts the reduction of Richard Karp used to show that the scheduling problem $1||\sum w_j(1-U_j)$ is \NP-hard~\cite{karp}.
While Karp gives a reduction from \KP, we use a reduction from $k$-\SubSum.
This result was also observed by~\cite{ResourceAwarePD,HartmannSteel2006} for \MPD with integer cost on taxa, which may be viewed as a special case of \tPDs or \tPDws.

\begin{proposition}[\cite{karp,ResourceAwarePD,HartmannSteel2006}]
	\label{prop:Knapsack}
	\begin{inparaenum}[(a)]
		\itemsep-.35em
		\item\label{thm:KP-NP}\tPDs and \tPDws are \NP-hard.
		\item\label{thm:KP-W1}It is \Wh{1}-hard  with respect to $k$ to decide whether an instance of \tPDs or \tPDws has a solution in which $k$ taxa are saved.
	\end{inparaenum}
	 Both statements hold even if the tree in the input is a star and $|T|=\var_{\ex}=1$.
\end{proposition}
\begin{proof}
	We reduce from $k$-\SubSum, 
	which is \NP-hard and \Wh{1}-hard when parameterized by the size $k$ of the solution~\cite{downey}.
	In $k$-\SubSum we are given a multiset of integers $\mathcal{Z} = \{z_1,\dots,z_{|\mathcal{Z}|}\}$ and integers $k,G\in \mathbb{N}$.
	It is asked whether there is a multiset $S\subseteq \mathcal{Z}$ of size $k$ such that $\sum_{z\in S} z = G$.
	
	\proofpara{Reduction}
	Let $(\mathcal{Z},G)$ be an instance $k$-\SubSum.
	Let $Q\in \mathbb{N}$ be a sufficiently large integer, i.e. bigger than the sum of all elements in $\mathcal{Z}$.
	
	We define an instance $\Instance = (\Tree, \ell, \ex, T, D)$ with
	a star $\Tree$ with root $\rho$ and leaves $X = \{ x_1,\dots,x_{|\mathcal{Z}|} \}$.
	For each $x_i\in X$, we set edge-weights $\w(\rho x_i)$ and 
	the rescue length $\ell(x_i)$ to be $z_i + Q$.
	Further, each taxon has extinction time $\ex(x_i) := G + kQ$ and
	the only team operates from $s_1 := 0$ to $e_1 := G + kQ$.
	Finally, we set $D := G + kQ$.
	
	\proofpara{Correctness}
	The reduction is clearly computed in polynomial time.
	Since $|T| = 1$, we have that $\Instance$ is a \yes-instance of \tPDs if and only if it is a \yes-instance of \tPDws.
	It remains to show that any solution for $\Instance$ saves exactly $k$ taxa, and that $\Instance$ is a \yes-instance of \tPDs if and only if $(\mathcal{Z},G)$ is a \yes-instance $k$-\SubSum.
	
	So let $(\mathcal{Z},G)$ be a \yes-instance of $k$-\SubSum with solution $S$.
	Define $S' := \{ x_i\in X \mid z_i \in S \}$. Clearly, $|S'| = |S| = k$.
	Then $\PD(S') = \sum_{x\in S'} \w(\rho x_i) = \sum_{x\in S'} (z_i + Q) = (\sum_{x\in S'} z_i) + kQ = G + kQ = D$, and
	analogously $\sum_{x\in S'} \ell(x_i) = \sum_{x\in S'} (z_i + Q) = G + kQ = H_{1}$.
	Subsequently, $S'$ is a solution for 
	\Instance.
	
	Let \Instance be a \yes-instance with solution $S'$.
	Because $G+kQ = D \le \PD(S') = \sum_{x_i\in S'} \w(\rho x_i) = \sum_{x_i\in S'} (z_i + Q) = \sum_{x_i\in S'} \ell(x_i) \le G+kQ$, we conclude that $G + kQ  = \sum_{x_i\in S'} (z_i + Q)$, from which it follows that $|S'| = k$ and $\sum_{x_i\in S'} z_i = G$. 
	Therefore $S := \{ z_i \mid x_i \in S' \}$ is a solution of $(\mathcal{Z},G)$.
\end{proof}

The scheduling variant \textsc{P3||C${}_{\max}$} in which it is asked if the given jobs can be scheduled on three parallel machines such that the \emph{make span} (the maximum time taken on any machine) is at most C${}_{\max}$, is strongly \NP-hard \cite{garey1978}.
\textsc{P3||C${}_{\max}$} can be seen as a special case of \textsc{P3||$\sum w_j(1-U_j)$} and therefore \tPDws in which there are three teams, the tree is a star on which all weights are 1, and we have to save every taxon in the time defined by the make span.

\begin{proposition}[\cite{garey1978}]
	\label{prop:ws-Dbar}
	\tPDws is strongly \NP-hard even if the tree is a star, every taxon has to be saved ($\Dbar=0$), there are three teams, $\max_\w=1$ and $\var_{\ex}=1$.
\end{proposition}

Our final observation in this section is that for an instance of \tPDs and a set of taxa $S\subseteq X$,
we can compute the diversity $\PD(S)$ and check whether there is a valid $T$-schedule saving $S$ in polynomial time (by \Cref{lem:scheduleCondition}).
Therefore, checking each subset of $X$ yields an $\Oh^*(2^{n})$-time-algorithm for \tPDs.
For an instance of \tPDws , we can check whether there is a strict valid $T$-schedule saving $S$ in $\Oh^*(|S|!)$ time, as follows.
For any ordering $x_1,\dots,x_{|S|}$ of the elements in $S$,
we can assign taxa $x_1,\dots,x_{i_1}$ to team $t_1\in T$ with $i_1$ being the biggest integer such that there is a valid $\{t_1\}$-schedule saving $\{x_1,\dots, x_{i_1}\}$.
For each $j \in \{2,\dots, |T|\}$, we assign the next $i_j$ available taxa in a similar way, choosing $i_j$ to be the largest integer such that there is a valid $\{t_j\}$-schedule saving $\{x_{i_{j-1}+1},\dots, x_{i_j}\}$. If all taxa are assigned to a team by this process, then there is a strict valid $T$-schedule saving $S$ (by \Cref{lem:strictScheduleCondition}).

\begin{proposition}
	\label{prop:X}
	\begin{inparaenum}[(a)]
		\itemsep-.35em
		\item\label{prop:X-s}\tPDs can be solved in $\Oh^*(2^{n})$ time and
		\item\label{prop:X-ws}\tPDws can be solved in $\Oh^*(2^{n} \cdot n!)$ time.
	\end{inparaenum}
\end{proposition}

\section{The Diversity $D$}
\label{sec:D}
In this section, we show that \tPDs and \tPDws are \FPT when parameterized by the threshold of diversity $D$.
When one tries to use the standard approach with a dynamic programming algorithm over the vertices of the tree, one will struggle with the question of how to handle the different extinction times of the taxa.
While it is straightforward how to handle taxa of a specific extinction time, already with a second, it is not trivial how to combine sub-solutions.
In order to overcome these issues, we use the technique of color-coding additionally to dynamic programming.

In the following, we consider colored versions of the problems, \ctPDs and \ctPDws.
In these, additionally in the input of the respective uncolored variant we are given a coloring~$c$ which assigns each edge $e\in E$ a set of colors $c(e)$: A subset of $[D]$ of size $\w(e)$.
For a taxon $x\in X$, we define $c(x)$ to be the union of colors on the edges between the root to~$x$. Further, for a set $S$ of taxa we define $c(S)$ to be $\bigcup_{x\in S} c(x)$.
The questions of \ctPDs (respectively, \ctPDws) is whether there is a subset $S$ of taxa and a (strictly) valid $T$-schedule saving $S$ such that $c(S) = [D]$.

Next, we show that \ctPDs and \ctPDws are \FPT with respect to~$D$.
The difficulty of combining sub-solutions for different extinction times also arises in these colored versions of the problem.
However, the color enables us to consider the tree as a whole.
We select taxa with ordered extinction time such that we are able to check whether there is a (strictly) valid $T$-schedule which can indeed be save the selected set of taxa.

\begin{lemma}
	\label{lem:cs-D}
	\ctPDs can be solved in $\Oh(2^D\cdot D\cdot \var_{\ex}^2 \cdot n \cdot \log(H_{\max_{\ex}}))$~time.
\end{lemma}
\begin{proof}
	Let \Instance be an instance of \ctPDs.
	For $p\in[\var_{\ex}]$, we call a set $S$ of taxa \emph{$p$-compatible} if $\ell(S\cap Z_q) \le H_i$ for each $q\in [p]$.
	For a set of colors $C\subseteq [D]$ and $p\in[\var_{\ex}]$, we call a set $S$ of taxa \emph{$(C,p)$-feasible} if
	\begin{enumerate}[F a)]
		\itemsep-.35em
		\item $C$ is a subset of $c(S)$,
		\item $S$ is a subset of $Z_p$, and
		\item $S$ is $p$-compatible.
	\end{enumerate}
	
	\proofpara{Table definition}
	In the following dynamic programming algorithm with table $\DP$, for each $C\subseteq [D]$ and $p\in [\var_{\ex}]$ we want the entry $\DP[C,p]$ to store the minimum length $\ell(S)$ of a $(C,p)$-feasible set of taxa $S\subseteq X$, with $\DP[C,p] = \infty$ if no $(C,p)$-feasible set exists.
	(The value $\infty$ could be replaced with a big integer. We calculate $\infty \pm z = \infty$ for each integer $z$.)

	\proofpara{Algorithm}
	As a base case we store $\DP[\emptyset, p] = 0$ for each $p\in [\var_{\ex}]$.
	
	To compute further values, we use the recurrence
	\begin{equation}
		\label{eqn:split-D}
		\DP[C,p] = 
		\min_{\substack{
				x\in Z_p,
				c(x)\cap C\ne \emptyset}}
		\psi(\DP[C\setminus c(x),\ex^*(x)] + \ell(x), H_{\ex^*(x)}).
	\end{equation}
	Recall $\ex^*(x) = q$ for $x\in Y_q$ and $\psi(a,b) = a$ if $a\le b$ and otherwise $\psi(a,b) = \infty$.
	
	We return \yes if $\DP[[D],\var_{\ex}] \le \max_{\ex}$ and otherwise we return \no.
	
	\proofpara{Correctness}
	Let \Instance be an instance of \ctPDs.
	From the definition of $(C,p)$-feasible sets it follows that there exists a $([D],\var_{\ex})$-feasible set $S$ of taxa if and only if $\Instance$ is a \yes-instance of \ctPDs.
	It remains to show that $\DP[C,p]$ stores the smallest length of a $(C,p)$-feasible set $S$.
	
	As the set $S=\emptyset$ is $(\emptyset,p)$-feasible for each $p\in [\var_{\ex}]$, the basic case is correct.
	
	As an induction hypothesis assume that for a fixed non-empty set of colors~$C\subseteq [D]$ and a fixed integer $p\in [\var_{\ex}]$, the entry $\DP[K,q]$ stores the correct value for each $K\subsetneq C$ and $q \le p$.
	We prove that $\DP[C,p]$ stores the correct value by showing 
	that if a $(C,p)$-feasible set of taxa $S$ exists then $\DP[C,p]\le \ell(S)$ and 
	that if $\DP[C,p] = a < \infty$ then there is a \emph{$(C,p)$-feasible} set $S$ with $\ell(S) = a$.
	
	Let $S\subseteq X$ be a $(C,p)$-feasible set of taxa.
	Observe that if $c(x)\cap C = \emptyset$ for a taxon $x\in S$ then also $S\setminus \{x\}$ is $(C,p)$-feasible.
	So we assume that $c(x)\cap C \ne \emptyset$ for each $x\in S$.
	Let $y\in S$ be a taxon such that $\ex(y) \ge \ex(x)$ for each $x\in S$.
	Observe that the set $S\setminus \{y\}$ is $(C\setminus c(y),\ex^*(y))$-feasible.
	Thus, because $y\in S\subseteq Z_p$ and $c(y)\cap C \ne \emptyset$, $\DP[C\setminus c(y),\ex^*(y)] \le \ell(S\setminus \{y\}) = \ell(S) - \ell(y)$ by the induction hypothesis.
	Because $S$ is $\ex^*(y)$-compatible, we conclude $\DP[C\setminus c(y),\ex^*(y)] + \ell(y) \le \ell(S) \le H_{\ex^*(y)}$.
	Hence, $\DP[C,p] \le \DP[C\setminus c(y),\ex^*(y)] + \ell(y) \le \ell(S)$.
	
	Secondly, suppose  $\DP[C,p] < \infty$.
	Then by the Recurrence~(\ref{eqn:split-D}) there is a taxon $x\in Z_p$ such that  $c(x)\cap C \ne \emptyset$ and $\DP[C,p] = \DP[C\setminus c(x),\ex^*(x)] + \ell(x)$.
	By the induction hypothesis there is a $(C\setminus c(x),\ex^*(x))$-feasible set $S'$ with $\DP[C\setminus c(x),\ex^*(x)] = \ell(S')$.
	Because $c(x)\cap(C\setminus c(x)) = \emptyset$, we may assume that $x\not\in S'$.
	Further, we conclude that $\DP[C,p] = \DP[C\setminus c(x),\ex^*(x)] + \ell(x) = \ell(S') + \ell(x) = \ell(S'\cup\{x\})$ and $S'\cup\{x\}$ is the desired $(C,p)$-feasible set.

	\proofpara{Running time}
	The table has $\Oh(2^D \cdot \var_{\ex})$ entries.
	For each $x\in X$ and every set of colors $C$ we can compute whether $c(x)$ and $C$ have a non-empty intersection in $\Oh(D)$ time.
	Then, we can compute the set $C\setminus c(x)$ in $\Oh(D)$ time.
	Each entry stores an integer that is at most $H_{\max_{\ex}}$ (or $\infty$).
	Thus the addition and the comparison in $\psi$ can be done in $\Oh(\log(H_{\max_{\ex}}))$ time, 
	and so the right side of Recurrence~(\ref{eqn:split-D}) can be computed in $\Oh(D \cdot \var_{\ex} \cdot n \cdot \log(H_{\max_{\ex}}))$ time.
	Altogether, we can compute a solution in $\Oh(2^D\cdot D\cdot \var_{\ex}^2 \cdot n \cdot \log(H_{\max_{\ex}}))$ time.
\end{proof}

For the algorithm in \Cref{lem:cs-D} it is crucial that 
the $(C,j)$-feasibility of a set $S$ depends only on $\ell(S)$ and the values $H_1,\dots, H_j$.
This is not the case for \ctPDws, as the available times of the teams impact a possible schedule.
Our solution for this issue is that we divide the colors (representing the diversity) that are supposed to be saved and delegate the colors to specific teams.
The problem then can be solved individually for the teams.
The division and delegation happens in the Recurrence~(\ref{eqn:ws-D}).

\begin{lemma}
	\label{lem:cws-D}
	\ctPDws can be solved in $\Oh^*(2^D\cdot D^3)$~time.
\end{lemma}
\begin{proof}
	\proofpara{Table definition}
	We define a dynamic programming algorithm with three tables $\DP_0, \DP_1$ and $\DP_2$.
	Entries of table $\DP_0$ are computed with the idea of \Cref{lem:cs-D}.
	$\DP_1$ is an auxiliary table.
	A final solution can be found in table $\DP_2$ then.
	For this lemma, similar to $H_i$, we define $H_i^{(q)} := |\{(i,j) \in \mathcal{H}_{T} \mid j\le \ex_q \}|$.
	That is, $H_i^{(q)}$ are the person-hours of team $t_i$ until $\ex_q$.
	
	Let \Instance be an instance of \ctPDws.
	We define terms similar to \Cref{lem:cs-D}.
	For a team $t_i\in T$ (and an integer $p\in[\var_{\ex}]$), a set $S\subseteq X$ is \emph{$t_i$-compatible} (respectively, \emph{$(p,t_i)$-compatible}) if $\ell(S\cap Z_q) \le H_i^{(q)}$ for each $q\in [\var_{\ex}]$ ($q\in [p]$).
	Inspired by \Cref{lem:strictScheduleCondition}, for a $i\in [|T|]$ we call a set $S$ of taxa \emph{$[i]$-compatible} if there is a partition $S_1,\dots,S_i$ of $S$ such that $S_j$ is $t_j$-compatible for each $j\in [i]$.
	For a set of colors $C\subseteq [D]$, an integer $p\in[\var_{\ex}]$ and a team $t_i \in T$, we call a set $S$ of taxa \emph{$(C,p,t_i)$-feasible} if
	\begin{enumerate}[F a)]
		\itemsep-.35em
		\item $S$ is a subset of $Z_p$,
		\item $S$ is $(p,t_i)$-compatible,
		\item $C$ is a subset of $c(S)$, and
		\item ($S=\emptyset$ or $S\cap Y_p \ne \emptyset$).
	\end{enumerate}
	
	Formally, for a set of colors $C\subseteq [D]$, an integer $p\in [\var_{\ex}]$ and a team $t_i\in T$ we want that the table entries store
	\begin{itemize}
		\itemsep-.35em
		\item $\DP_0[C,p,i]$: The minimum length $\ell(S)$ of a $(C,p,t_i)$-feasible set of taxa $S\subseteq X$, and
		\item $\DP_1[C,i]$ (respectively $\DP_2[C,i]$): 1 if there is a $t_i$-compatible ($[i]$-compatible) set of taxa $S\subseteq X$ with $c(S) \supseteq C$ and 0 otherwise.
	\end{itemize}
	
	\proofpara{Algorithm}
	We store $\DP_0[\emptyset, p, i] = 0$ for each $p\in [\var_{\ex}]$ and $i\in [|T|]$.
	
	To compute further values of $\DP_0$, we use the recurrence
	\begin{equation}
	\DP_0[C,p,i] = \min_{x\in Y_p, c(x)\cap C\ne \emptyset, q\le p}
	\psi(\DP[C\setminus c(x),q,i] + \ell(x), H_i^{(q)}).
	\end{equation}
	Recall $\psi(a,b) = a$ if $a\le b$ and otherwise $\psi(a,b) = \infty$.
	Once all values in $\DP_0$ are computed, we compute $z_{C,i} := \min_{q\in [\var_{\ex}]} \DP[C,q,i]$
	and set $\DP_1[C,i] = 1$ if $z_{C,i} \le H_i^{(\max_{\ex})}$ and $\DP_1[C,i] = 0$ otherwise.
	
	Finally, we define $\DP_2[C,1] := \DP_1[C,1]$ and compute further values by
	\begin{equation}
	\label{eqn:ws-D}
	\DP_2[C,i+1] = \min_{C' \subseteq C} \DP_2[C',i] \cdot \DP_1[C\setminus C',i+1].
	\end{equation}
	
	We return \yes if $\DP_2[[D],|T|] = 1$ and otherwise we return \no.
	
	\proofpara{Correctness}
	Analogously as in \Cref{lem:cs-D}, one can prove that the entries of $\DP_0$ store the correct value.
	It directly follows by the definition that the entries of $\DP_1$ store the correct value.
	Likewise $\DP_2[C,1]$ stores the correct value for colors $C\subseteq [D]$ be the definition.
	
	Fix a set of colors $C\subseteq [D]$ and an integer $i\in [|T|-1]$.
	We assume as induction hypothesis that for each $K\subseteq C$ and $j\le i$ the entry $\DP_2[K,j]$ stores the correct value.
	To conclude that $\DP_2[C,j+1]$ stores the correct value, we show that 
		$\DP_2[C,i+1] \le \ell(S)$ for each $[i+1]$-compatible set $S\subseteq X$ with $c(S) \supseteq C$ and
		if $\DP_2[C,i+1] = a < \infty$ then there is an $[i+1]$-compatible set $S\subseteq X$ with $c(S) \supseteq C$ and $\ell(S) = a$.
	
	Let $S\subseteq X$ be an $[i+1]$-compatible set with $c(S) \supseteq C$.
	Let $S_1,\dots,S_{i+1}$ be a partition of $S$ such that $S_j$ is $t_j$-compatible for each $j\in [i+1]$.
	Define $\hat S := \bigcup_{j=1}^i S_j$ and $\hat C := C \cap c(\hat S)$.
	Then $\hat S$ is $[i]$-compatible with $c(\hat S) \supseteq \hat C$.
	Therefore $\DP_2[\hat C,i] \le \ell(\hat S)$ by the induction hypothesis.
	Because $c(S_{i+1}) \supseteq C \setminus \hat C$
	we conclude that $\DP_1[C\setminus \hat C,i+1] \le \ell(S_{i+1})$.
	Thus, $\DP_2[C,i+1] \le \DP_2[\hat C,i] + \DP_1[C\setminus \hat C,i+1] \le \ell(\hat S) + \ell(S_{i+1}) = \ell(S)$.
	
	Let $\DP_2[C,i+1]$ store the value $a<\infty$.
	By Recurrence~(\ref{eqn:ws-D}), there is a set of colors $C'\subseteq C$ such that $a = \DP_2[C',i] + \DP_1[C\setminus C',i+1]$.
	By the induction hypothesis we conclude then that there is
	an $[i]$-compatible set $S_1$ and
	a $t_{i+1}$-compatible set $S_2$
	such that $\DP_2[C',i] = \ell(S_1)$ with $c(S_1)\supseteq C'$ and $\DP_1[C\setminus C',i+1] = \ell(S_2)$ with $c(S_2)\supseteq C\setminus C'$.
	Assume first that $S_1$ and $S_2$ are disjoint.
	We define $S := S_1 \cup S_2$ and conclude with the disjointness that $\ell(S) = \ell(S_1) \cup \ell(S_2) = a$.
	Further, we conclude that $S$ is $[i+1]$-compatible and $c(S) = c(S_1) \cup c(S) = C' \cup (C\setminus C') = C$
	such that $S$ is the desired set.
	It remains to argue that non-disjoint $S_1$ and $S_2$ contradict the optimality of $C'$.
	
	Now let $A$ be the intersection of $S_1$ and $S_2$ and further let $\hat C := c(S_1)$ and $\hat S := S_2 \setminus A$.
	Observe that $\hat S$ is a $t_{i+1}$-compatible set with $c(\hat S) \supseteq (C\setminus C') \setminus c(A) = C \setminus \hat C$.
	Therefore $\DP_2[\hat C,i] + \DP_1[C\setminus \hat C,i+1] \le \ell(S_1) + \ell(\hat S) = \ell(S_1) + \ell(S_2) - \ell(A)  = \DP_2[C',i] + \DP_1[C\setminus C',i+1] - \ell(A)$,
	which unless $A=\emptyset$ contradicts the optimallity of $C'$.

	\proofpara{Running time}
	By \Cref{lem:cs-D}, all entries of $\DP_0$ can be computed in $\Oh(2^D\cdot D\cdot \var_{\ex}^2 \cdot n \cdot \log(H_{\max_{\ex}}))$ time.
	Table $\DP_1$ contains $2^D \cdot |T|$ entries which can be computed in $\Oh(\var_{\ex} \cdot \log(H_{\max_{\ex}}))$ time each.
	
	To compute the values of table $\DP_2$, we use convolutions.
	Readers unfamiliar to this topic, we refer to \cite[Sec. 10.3]{cygan}.
	We define functions $f_i, g_i: 2^{[D]} \to \{0,1\}$ where
	$f_i(C) := \DP_2[C,i]$ and $g_i(C) := \DP_1[C,i]$.
	Then we can express Recurrence~(\ref{eqn:ws-D}) as $\DP_2[C,i+1] = f_{i+1}(C) = (f_i * g_{i+1})(C)$.
	Therefore, for each $i\in [|T|]$ we can compute all values of $\DP_2[\cdot,i] = f_i(\cdot)$ in $\Oh(2^D \cdot D^3)$ time~\cite[Thm. 10.15]{cygan}.
	
	Altogether, in $\Oh(2^D\cdot D^3 \cdot |T| \cdot \var_{\ex}^2 \cdot n \cdot \log(H_{\max_{\ex}}))$ time, we can compute a solution of an instance of \ctPDws.
\end{proof}

Our following procedure
in a nutshell
is to use standard color coding techniques to show that in $\Oh^*(2^{\Oh(D)})$~time
one can reduce an instance $\Instance$ of \tPDs (respectively, \tPDws) to $z \in \Oh^*(2^{\Oh(D)})$ instances $\Instance_1,\dots,\Instance_z$ of \ctPDs (\ctPDws),
which can be solved using \Cref{lem:cs-D} and \Cref{lem:cws-D}.
Then \Instance is a \yes-instance, if and only if any $\Instance_i$ is a \yes-instance.

\begin{theorem}
	\label{thm:D}
	\tPDs and \tPDws can be solved in 
	$\Oh^*(2^{2.443\cdot D + o(D)})$~time.
\end{theorem}

For an overview of color coding, we refer the reader to~\cite[Sec.~5.2~and~5.6]{cygan}.

The key idea is that we construct a family $\mathcal{C}$ of colorings on the edges of $\Tree$, where each edge $e\in E$ is assigned a subset of $[D]$ of size $\w(e)$.
Using these, we generate~$|\mathcal{C}|$ instances of the colored version of the problem, which we then solve in $\Oh^*(2^D\cdot D)$ time.
Central for this will be the concept of a perfect hash family.

\begin{definition}
	\label{def:hash-fam}
	For integers $W$ and $D$,
	a~\emph{$(W,D)$-perfect hash family $\mathcal{F}$} is a family of functions $f: [W] \to [D]$ such that for every subset $Z$ of $[W]$ of size~$[D]$, some $f \in \mathcal{F}$ exists such that $|f^{-1}(i)\cap Z| = 1$ for each $i \in [D]$.
\end{definition}

\begin{proposition}[\cite{Naor,cygan}]
	For any integers $W, D \geq 1$ one can construct a $(W, D)$-perfect hash family of size $e^D D^{\Oh(\log D)} \cdot \log W$ in time $e^D D^{\Oh(\log D)} \cdot W \log W$.
\end{proposition}

\begin{proofm}{of \Cref{thm:D}}
	\proofpara{Reduction}
	We focus on \tPDs and omit the analogous proof for \tPDws.
	Let $\Instance = (\Tree, r, \ell, T, D)$ be an instance of \tPDs.
	We assume for each taxon $x\in X$ there is a valid $T$-schedule saving $\{x\}$, as we can delete $x$ from \Tree, otherwise.
	Therefore, if there is an edge $e$ with $\w(e) \ge D$, then $\{x\}$ is a solution for each $x\in \off(e)$ and we have a trivial \yes-instance.
	Hence, we assume that $\max_\w < D$.
	
	Now, order the edges $e_1, \dots, e_m$ of $\Tree$ arbitrarily and
	define $W_0 := 0$ and $W_j := \sum_{i=1}^{j} \w(e_{i})$ for each $j\in [m]$.
	We set $W := W_m$.
	Let $\mathcal{F}$ be a $(W, D)$-perfect hash family.
	Now we define the family of colorings $\mathcal{C}$ as follows.
	For every $f \in \mathcal{F}$, let $c_f$ be the coloring that assigns edge $e_j$ the colors $\{f(W_{j-1}+1), \dots, f(W_j)\}$.
	
	For each $c_f \in \mathcal{C}$,
	let $\Instance_{c_f} = (\Tree, r, \ell, T, D, c_f)$ be the corresponding instance of \ctPDs.
	Now solve each $\Instance_{c_f}$ using the algorithm of \Cref{lem:cws-D}, and return \yes if and only if $\Instance_{c_f}$ is a \yes-instance for some $c_f \in \mathcal{C}$.
	
	\proofpara{Correctness}
	For any subset of edges $E'$ with $\w(E') \geq D$, there is a corresponding subset of $[W]$ of size at least $D$.
	Since $\mathcal{F}$ is a $(W, D)$-perfect hash family, it follows that $c_f(E') = [D]$, for some $c_f \in \mathcal{C}$.
	Thus in particular, if $S\subseteq X$ is a solution for $\Instance$, then $c_f(S) = [D]$, for some $c_f \in \mathcal{C}$, where $c_f(S)$ is the colors assigned by $c_f$ to the edges between the root and a taxon of $S$.
	It follows that one of the constructed instances of \ctPDs is a \yes-instance.
	
	Conversely, it is easy to see that any solution $S$ for one of the constructed instances of \ctPDs is also a solution for $\Instance$.

	\proofpara{Running Time}
	The construction of $\mathcal{C}$ takes $e^D D^{\Oh(\log D)} \cdot W \log W$ time, and for each $c \in \mathcal{C}$ the construction of each instance of \ctPDs takes time polynomial in $|\Instance|$. Solving each instance of \ctPDs takes $\Oh^*(2^D\cdot D)$~time, and the number of instances is $|\mathcal{C}| = e^D D^{\Oh(\log D)} \cdot \log W$.
	
	Thus, the total running time is 
	$\Oh^*(e^D D^{\Oh(\log D)} \log W \cdot (W + (2^D\cdot D)))$.
	Because $W = \PD(X) < 2n\cdot D$ this simplifies to $\Oh^*((2e)^D \cdot 2^{\Oh(\log^2(D))})$.
\end{proofm}

\section{The acceptable loss of diversity \Dbar}
\label{sec:Dbar}
In this section we show that \tPDs is \FPT with respect to the acceptable loss of phylogenetic diversity \Dbar, which is defined as $\PD(X)-D$.
For a set of taxa~$A \subseteq X$, define $E_d(A) := \{ e\in E \mid \off(e)\subseteq A \} = E \setminus E_\Tree(X\setminus A)$.
That is, $E_d(A)$ is the set of edges that \emph{do not} have offspring in $X\setminus A$. 
Then one may think of \tPDs as the problem of finding a set of taxa $S$ such that there is a valid $T$-schedule saving $S$ and $\w(E_d(X\setminus S))$ is at most $\Dbar$.

In order to show the desired result, we again use the techniques of color-coding and dynamic programming.
For an easier understanding of the following definitions consider \Cref{fig:example-Dbar-definitions}.
\begin{figure}[t]
	\centering
	\resizebox{1.1\columnwidth}{!}{%
	\begin{tikzpicture}[scale=1,every node/.style={scale=0.9}]
			\tikzstyle{txt}=[circle,fill=white,draw=white,inner sep=0pt]
			\tikzstyle{nde}=[circle,fill=black,draw=black,inner sep=2.5pt]
			
			\node[nde] (v0) at (5,10) {};
			\node[nde] (v1) at (2,9) {};
			\node[nde] (v2) at (5,9) {};
			\node[nde] (v3) at (8,9) {};
			\node[nde] (v4) at (1.6,8) {};
			\node[nde] (v5) at (2.4,8) {};
			\node[nde] (v6) at (4.6,8) {};
			\node[nde] (v7) at (5.4,8) {};
			\node[nde] (v8) at (7.6,8) {};
			\node[nde] (v9) at (8.4,8) {};
			
			\node[txt,xshift=9mm,yshift=-10mm] (c1) [above=of v1] {$e_1$};
			\node[txt,xshift=4mm,yshift=-10mm] (c2) [above=of v2] {$e_2$};
			\node[txt,xshift=-9mm,yshift=-10mm] (c3) [above=of v3] {$e_3$};
			
			\node[txt,xshift=-2mm,yshift=-10mm] (c4) [above=of v4] {$e_4$};
			\node[txt,xshift=2mm,yshift=-10mm] (c5) [above=of v5] {$e_5$};
			\node[txt,xshift=-2mm,yshift=-10mm] (c6) [above=of v6] {$e_6$};
			\node[txt,xshift=2mm,yshift=-10mm] (c7) [above=of v7] {$e_7$};
			\node[txt,xshift=-2mm,yshift=-10mm] (c8) [above=of v8] {$e_8$};
			\node[txt,xshift=2mm,yshift=-10mm] (c9) [above=of v9] {$e_9$};
			
			\node[txt,yshift=9mm] (r4) [below=of v4] {$15$};
			\node[txt,yshift=9mm] (r5) [below=of v5] {$30$};
			\node[txt,yshift=9mm] (r6) [below=of v6] {$25$};
			\node[txt,yshift=9mm] (r7) [below=of v7] {$15$};
			\node[txt,yshift=9mm] (r8) [below=of v8] {$25$};
			\node[txt,yshift=9mm] (r9) [below=of v9] {$25$};
			
			\node[txt,yshift=9mm] (l4) [below=of r4] {$10$};
			\node[txt,yshift=9mm] (l5) [below=of r5] {$13$};
			\node[txt,yshift=9mm] (l6) [below=of r6] {$9$};
			\node[txt,yshift=9mm] (l7) [below=of r7] {$7$};
			\node[txt,yshift=9mm] (l8) [below=of r8] {$9$};
			\node[txt,yshift=9mm] (l9) [below=of r9] {$12$};
			
			\node[txt,xshift=3mm] (l'4) [left=of l4] {$\ell(x)$};
			\node[txt,xshift=3mm] (r'4) [left=of r4] {$\ex(x)$};
			
			\node[txt,xshift=-9mm] (t0) [right=of v0] {$v_0$};
			\node[txt,xshift=9mm] (t1) [left=of v1] {$v_1$};
			\node[txt,xshift=-9mm] (t2) [right=of v2] {$v_2$};
			\node[txt,xshift=-9mm] (t3) [right=of v3] {$v_3$};
			\node[txt,xshift=9mm] (t4) [left=of v4] {$x_1$};
			\node[txt,xshift=-9mm] (t5) [right=of v5] {$x_2$};
			\node[txt,xshift=9mm] (t6) [left=of v6] {$x_3$};
			\node[txt,xshift=-9mm] (t7) [right=of v7] {$x_4$};
			\node[txt,xshift=9mm] (t8) [left=of v8] {$x_5$};
			\node[txt,xshift=-9mm] (t9) [right=of v9] {$x_6$};
			
			\draw[thick,dashed,blue,arrows = {-Stealth[length=8pt]}] (v0) to (v1);
			\draw[thick,arrows = {-Stealth[length=8pt]}] (v0) to (v2);
			\draw[thick,arrows = {-Stealth[red,length=8pt]}] (v0) to (v3);
			\draw[thick,dashed,blue,arrows = {-Stealth[length=8pt]}] (v1) to (v4);
			\draw[thick,dashed,blue,arrows = {-Stealth[red,length=8pt]}] (v1) to (v5);
			\draw[thick,arrows = {-Stealth[length=8pt]}] (v2) to (v6);
			\draw[thick,arrows = {-Stealth[length=8pt]}] (v2) to (v7);
			\draw[thick,arrows = {-Stealth[red,length=8pt]}] (v3) to (v8);
			\draw[thick,dashed,blue,arrows = {-Stealth[length=8pt]}] (v3) to (v9);
			
			\node[txt] at (11,8.6) {
				\begin{tabular}{c|cc}
					$e$ & $\hat c(e)$ & $c^-(e)$ \\
					\hline
					$e_1$ & 2 & $\{1,7\}$ \\
					$e_2$ & 9 & $\{2,3\}$ \\
					$e_3$ & 12 & $\in E_b$ \\
					$e_4$ & 10 & $\{9\}$ \\
					$e_5$ & 6 & $\{4,11\}$ \\
					$e_6$ & 3 & $\{7\}$ \\
					$e_7$ & 8 & $\in E_b$ \\
					$e_8$ & 4 & $\emptyset$ \\
					$e_9$ & 5 & $\{3,8\}$ \\
				\end{tabular}
					};
		\end{tikzpicture}
	}
		
	\caption{A hypothetical extinction-6-colored $X$-tree.
			The anchored taxa set $\mathcal A := \{(x_1,v_1,e_5), (x_2,v_0,e_3), (x_6,v_3,e_8)\}$
			with $c(E^+(\mathcal A)) = c(\{ e_1, e_4, e_5, e_9 \}) = [11]$ 
			and $\hat c(E_s(\mathcal A)) = \hat c(\{ e_3, e_5, e_8 \}) = \{ 4, 6, 12 \}$
			is color-respectful.
			The edges in $E^+(\mathcal A)$ are blue and dashed. The edges in $E_s(\mathcal A)$ have a red arrow.
			The anchored taxa set $\mathcal A' := \{(x_1,v_0,e_3), (x_2,v_1,e_4), (x_6,v_3,e_8)\}$
			does not have a valid ordering.
			}
	\label{fig:example-Dbar-definitions}
\end{figure}
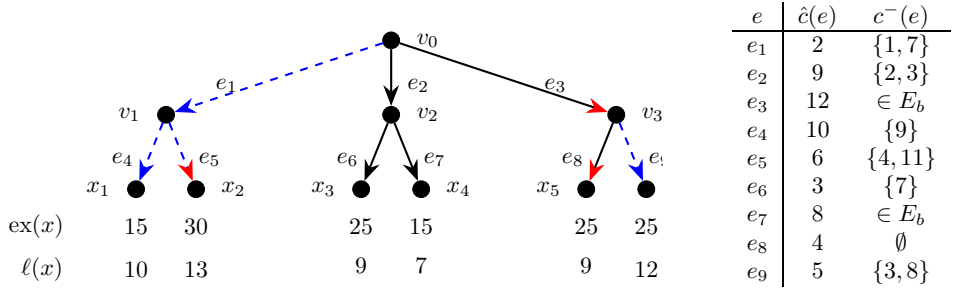%

We start by defining a colored version of the problem.
With $E_{\le \Dbar}$ (or $E_{> \Dbar}$) we denote the set of edges $e\in E$ with $\w(e)\le \Dbar$ (or $\w(e) > \Dbar$).
An \emph{extinction-$\Dbar$-colored} $X$-tree is a directed binary $X$-tree $\Tree = (V,E,\w,\hat c,c^-)$,
where $\hat c$ assigns each edge $e\in E$ a \emph{key-color} $\hat c(e) \in [2\Dbar]$ and
$c^-$ assigns edges $e\in E_{\le \Dbar}$ a set of colors $c^-(e) \subseteq [2\Dbar]\setminus \{\hat c(e)\}$ of size $\w(e)-1$.
With $c(e)$ we denote the union of the two sets $c^-(e) \cup \{\hat c(e)\}$ for each edge $e\in E_{\le \Dbar}$.

Observe that while in \Cref{sec:D} we wanted the set of edges with an offspring in $S$ to use every color at least once, here we want that the edges in $E_d(S)$ use \emph{at most} a certain number of colors.

Note that the set of colors $c(E_d(A))$ cannot be determined strictly from the sets $c(E_d(\{x\}))$ for each $x \in X$.
In particular, if two taxa $x_1$ and $x_2$ do not share a parent then $E_d(\{x_1,x_2\}) = E_d(\{x_1\}) \cup E_d(\{x_2\})$,
but if $x_1$ and $x_2$ are the only two children of $v$ then $E_d(\{x_1,x_2\})$ also contains the parent-edge of $v$, which is not in $E_d(\{x_1\})$ or $E_d(\{x_2\})$.
This presents a challenge when designing a dynamic programming algorithm that adds taxa to a partial solution one taxon at a time.
To get around this, we introduce the following concept.
An \emph{anchored taxa set} $\mathcal A$ is a set of tuples $(x,v,e)$, where $x\in X$ is a taxon, $v$ is a strict ancestor of $x$ and $e$ is an outgoing edge of $v$ with $x\not\in \off(e)$.
In the rest of the section whenever we refer to a tuple $(x,v,e)$, we always assume $x\in \off(v)$ and $e$ is an outgoing edge of $v$ with $x\not\in \off(e)$.
We denote with $X(\mathcal A) := \{ x \mid (x,v,e) \in \mathcal A \}$ the taxa of an anchored taxa set.
For an $X$-tree \Tree, a taxon $x\in X$ and a vertex $v \in \anc(x)$, we define $\Pvx$ to be the set of edges on the path between~$v$ and~$x$.

For an anchored taxa set $\mcA$,
we define edge sets $E^+(\mcA) := \bigcup_{(x,v,e) \in \mcA} \Pvx$ and
$E_s(\mcA) := \{ e \mid (x,v,e) \in \mcA \}$.
Informally, $E^+(\mcA)$ is the set of edges that connect the taxa of the anchored taxa set with the anchors and
the edges in $E_s(\mcA)$ are sibling-edges of the topmost edges of $\Pvx$ for each $(x,v,e) \in \mcA$. These sibling edges may or may not be in $E_d(X(\mcA))$, depending on whether $e$ is part of $P_{u,y}$ for some other $(y,u,e') \in \mcA$.

A set of edges $E'$ has \emph{unique colors} (or \emph{unique key-colors}), if the sets 
$c(e_1)$ and $c(e_2)$ (respectively, $\{\hat c(e_1)\}$ and $\{\hat c(e_2)\}$)
are disjoint for each tuple $e_1,e_2\in E'$, with $e_1 \ne e_2$.
We define $\mcA_j := \{ t_i \in \mcA \mid i\le j \}$.
An anchored taxa set $\mcA$ has a \emph{valid ordering $(x_1,v_1,e_1),\dots,(x_{|\mcA|},v_{|\mcA|},e_{|\mcA|})$} if $(x_i,v_i,e_i) \in \mcA$, $\ex(x_i) \le \ex(x_j)$ and $\hat c(e_j) \not\in c(E^+(\mcA_j))$ for each pair $i,j\in [|\mcA|]$, $i\le j$.
An anchored taxa set $\mcA$ is \emph{color-respectful} if 
\begin{enumerate}[CR a)]
	\itemsep-.35em
	\item\label{it:Ca} $E^+(\mcA)$ has unique colors,
	\item\label{it:Cb} $E_s(\mcA)$ has unique key-colors,
	\item\label{it:Cc} $E^+(\mcA)$ and $E_{> \Dbar}$ are disjoint,
	\item\label{it:Cd} $\mcA$ has a valid ordering, and
	\item\label{it:Ce} $\Pvx$ and $P_{u,y}$ are disjoint for any $(x,v,e),(y,u,e') \in \mcA$.
\end{enumerate}

The existence of a color-respectful anchored taxa set will be used to show that an instance of \tPDs is a \yes-instance. To formally show this, we first define a colored version of the problem.

\problemdef{\cBartPDslong\Dbar (\cBartPDs\Dbar)}{
	An extinction-$\Dbar$-colored $X$-tree $\Tree = (V,E,\w,\hat c,c^-)$,
	integers $\ell(x)$ and $\ex(x)$ for each $x\in X$,
	and teams $T$.
}{
	Is there an anchored taxa set $\mcA$ such that $\mcA$ is color-respectful, $|c(E^+(\mcA))| \le \Dbar$ and there is a valid $T$-schedule saving $X\setminus X(\mcA)$?
}

The following lemma shows how \cBartPDs\Dbar is relevant to \tPDs.

\begin{lemma}\label{lem:coloredYes}
	Let $\Instance' = (\Tree', r, \ell, T, D)$ with $\Tree' = (V,E,\w,\hat c,c^-)$, be an instance of \cBartPDs\Dbar
	and let $\Instance = (\Tree, r, \ell, T, D)$ with $\Tree = (V,E,\w)$ be the instance of \tPDs induced by $\Instance'$ when omitting the coloring.
	If $\Instance'$ is a \yes-instance of \cBartPDs\Dbar, then $\Instance$ is a \yes-instance of \tPDs. 
\end{lemma}
\begin{proof}
	Let $\mcA$ be a solution for $\Instance'$. As there exists a $T$-schedule saving the set of taxa $S := X\setminus X(\mcA)$, it is sufficient to show that $\PD(S) \ge D$.
	 
	Since every edge in $\Tree$ either has an offspring in $S$ or is in $E_d(X\setminus S)$, it follows that $\w(E_d(X\setminus S)) = \PD(X) - \PD(S)$. Thus in order to show $\PD(S) \ge D$ it  remains to show that $\w(E_d(X\setminus S)) \leq \PD(X) - D = \Dbar$.
	
	Let $(x_1,v_1,e_1),\dots,(x_{|\mcA|},v_{|\mcA|},e_{|\mcA|})$ be a valid ordering of $\mcA$ and let 
	$\mcA_j := \{ (x_i,v_i) \in \mcA \mid i\le j \}$ for each $j \in [|\mcA|]$.
	We prove by induction on $j$ that $E_d(X(\mcA_j))\subseteq E^+(\mcA_j)$ for each $j \in [|\mcA|]$.
	 
	For the base case $j= 1$, we have that $X(\mcA_1) = \{x_1\}$ and $E_d(\{x_1\})$ consists of the single incoming edge of $x$. As this edge is part of $P_{v_1,x_1}$, the claim holds.
	 
	For the induction step, assume that $j >1$ and  $E_d(X(\mcA_{j-1}))\subseteq E^+(\mcA_{j-1})$. 
	It is sufficient to show that $E_d(X(\mcA_j)) \setminus E_d(X(\mcA_{j-1})) \subseteq P_{v_j,x_j}$, as $E^+(\mcA_j) =  E^+(\mcA_{j-1}) \cup P_{v_j,x_j}$.
	To see this, consider the edge $e_j$.
	Because $\mcA$ has a valid ordering, $\hat c(e_j) \not\in c(E^+(\mcA_j))$, which implies $e_j \notin E^+(\mcA_j)$.
	By the inductive hypothesis, $E_d(X(\mcA_{j-1}))\subseteq E^+(\mcA_{j-1}) \subseteq E^+(\mcA_j)$ and so $e_j \notin E_d(X(\mcA_{j-1}))$.
	Thus $e_j$ has an offspring $z$ not in $X(\mcA_{j-1})$.
	By definition $x_j \notin \off(e_j)$, so we conclude that $z \notin X(\mcA_{j-1}) \cup \{x_j\} = X(\mcA_j)$. 
	Furthermore $z \in \off(e')$ for any edge $e'$ above $v_j$, and so $e' \notin E_d(X(\mcA_j))$ for any such $e'$.
	It follows that the only edges in $E_d(X(\mcA_j)) \setminus E_d(X(\mcA_{j-1}))$ must be between $x_j$ and $v_j$, that is, in $P_{v_j,x_j}$.
	We conclude $E_d(X(\mcA_j))\subseteq E_d(X(\mcA_{j-1}))\cup P_{v_j,x_j} \subseteq E^+(\mcA_{j-1}) \cup P_{v_j,x_j} = E^+(\mcA_j)$.
	
	Letting $j = |\mcA|$, we have $E_d(X(\mcA)) \subseteq E^+(\mcA)$.
	This implies $\w(E_d(X(\mcA))) \leq \w(E^+(\mcA)) = \sum_{e \in E^+(\mcA)}|c(e)| = |c(E^+(\mcA))|$, where the last equality holds because  $E^+(\mcA)$ has unique colors.
	We have $\w(E_d(X(\mcA))) \leq \Dbar$ since $|c(E^+(\mcA))|\leq \Dbar$ and so $\PD(S) \ge D$, as required.
\end{proof}

\medskip

In the following,
we continue to define a few more things.
Then, we present Algorithm~$(\Dbar)$ for solving \cBartPDs\Dbar.
Afterward, we analyze the running time of Algorithm~$(\Dbar)$.
Finally we show how to reduce instances of \tPDs to instances of \cBartPDs\Dbar to conclude the desired \FPT-result.
\medskip

Define $\Hbar{p} := \ell(Z_p) - H_p$ to be the number of person-hours one would need additionally to be able to save all taxa in $Z_p$ (not regarding the time constraints $\ex_1,\dots,\ex_{p-1}$).
We call a set of taxa $A\subseteq X$  \emph{ex-$q$-compatible} for $q\in[\var_{\ex}]$ if $\ell(A\cap Z_p) \ge \Hbar{p}$ for each $p\in [q-1]$. 
Note that in order for there to be a valid $T$-schedule saving $S\subseteq X$, it must hold that $\ell(S\cap Z_p) \leq H_p$, and so $\ell((X\setminus S)\cap Z_p) \geq \Hbar{p}$, for each $p \in[\var_{\ex}]$.
Thus, if $A$ is ex-$q$-compatible then is a valid $T$ schedule saving $(X\setminus A)\cap Z_{q-1}$ (but not necessarily $(X\setminus A)\cap Z_{q}$).

For our dynamic programming we will need to track the existence of color-respectful anchored taxa sets using particular sets of colors. To this end we define the notion of \emph{ex-$(C_1,C_2,q)$-feasibility}.
For sets of colors $C_1,C_2\subseteq [2\Dbar]$ and an integer $q\in[\var_{\ex}]$, an anchored taxa set $\mathcal A$ is \emph{ex-$(C_1,C_2,q)$-feasible} if 
\begin{enumerate}[F a)]
	\itemsep-.35em
	\item\label{it:Fa} $\mcA$ is color-respectful,
	\item\label{it:Fb} $c(E^+(\mcA))$ is a subset of $C_1$,
	\item\label{it:Fc} $\hat c(e)$ is in $C_2 \cup c(E^+(\mcA \setminus \{(x,v,e)\}))$ for each $(x,v,e)\in \mcA$,
	\item\label{it:Fd} $X(\mcA)$ is a subset of $Z_q$, and
	\item\label{it:Fe} $X(\mcA)$ is ex-$q$-compatible.
\end{enumerate}
As an example, observe that the anchored taxa set $\mathcal A$ in \Cref{fig:example-Dbar-definitions} is ex-$(C_1,C_2,q)$-feasible for $C_1 = [11]$, $C_2 = \{12\}$ and $q=3$ if and only if $\mathcal A$ is $q$-compatible (which we don't know as the teams are not given). This is the case if and only if $\Hbar{1} \ge 10$, $\Hbar{2} \ge 22$ and $\Hbar{3} \ge 35$.

\medskip

Our algorithm calculates for each combination $C_1,C_2$, and $q$ the maximum length of an ex-$(C_1,C_2,q)$-feasible anchored taxa set.
In order to calculate these values recursively, we declare which tuples $(x,v,e)$ are suitable for consideration at each stage in the algorithm.

A tuple $(x,v,e)$ is \emph{$(C_1,C_2)$-good} for sets of colors $C_1,C_2\subseteq [2\Dbar]$ if
\begin{enumerate}[G a)]
	\itemsep-.35em
	\item\label{it:Ga} \Pvx has unique colors,
	\item\label{it:Gb} $c(\Pvx)$ is a subset of $C_1$, 
	\item\label{it:Gc} $\hat{c}(e)$ is in $C_2$, and 
	\item\label{it:Gd} $\Pvx$ and $E_{> \Dbar}$ are disjoint.
\end{enumerate}

Let $\mathcal X_{(q)}$ be the set of tuples $(x,v,e)$ such that $x\in Z_q$, $\Pvx$ has unique colors, $\Pvx \subseteq E_{\le \Dbar}$ and $\hat c(e) \not\in c(\Pvx)$.
Disjoint sets of colors $C_1,C_2\subseteq [2\Dbar]$ are \emph{$q$-grounding},
if $c(\Pvx) \not\subseteq C_1$ or $\hat c(e) \not\in C_2$ for each tuple $(x,v,e)\in \mathcal X_{(q)}$.
That is, $C_1$ and $C_2$ are $q$-grounding if and only if there is no $(C_1,C_2)$-good tuple $(x,v,e)$ with $x \in Z_q$.
In \Cref{fig:example-Dbar-definitions} the set $\mathcal X_{(1)}$ would be $\{(x_1,v_1,e_5),(x_1,v_0,e_3)\}$---both paths $P_{v_0,x_4}$ and $P_{v_1,x_4}$ contain the edge $e_7\in E_{>\Dbar}$, and $\hat c(e_2) = 9 \in c(e_4)$ so $(x_1,v_0,e_2)$ is not contained in $\mathcal X_{(1)}$.
Therefore, the 1-grounding colors are any disjoint sets of colors $C_1$ and $C_2$ with
($c(e_4) \not\subseteq C_1$ or $\hat c(e_5) = 6 \not\in C_2$)
and
($c(\{e_1,e_4\}) \not\subseteq C_1$ or $\hat c(e_3) = 12 \not\in C_2$).
That is the case if $C_1$ or $C_2$ are empty.
One non-trivial example are the sets $C_1 = \{6,7,\dots,11\}$ and~$C_2 = [5] \cup \{12\}$.

\medskip

\medskip
\noindent\textbf{Algorithm~$\mathbf{(\Dbar)}$.}\\
Let \Instance be an instance of \cBartPDs\Dbar.
In the following dynamic programming algorithm, 
we compute a value $\DP[C_1,C_2,q]$ for all disjoint $C_1, C_2 \subseteq [2\Dbar]$ with $C_1\cap C_2 = \emptyset$ and $|C_1|\leq \Dbar$, and all $q \in [\var_{\ex}]$.
The entries of $\DP$ are computed 
in some order such that $\DP[C_1',C_2',q']$ is computed before $\DP[C_1'',C_2'',q'']$ if $C_1' \subsetneq C_1''$.
We want that $\DP[C_1,C_2,q]$ stores the \emph{maximum} rescue length $\ell(X(\mcA))$ of an ex-$(C_1,C_2,q)$-feasible anchored taxa set $\mcA$.
If there exists no ex-$(C_1,C_2,q)$-feasible anchored taxa set, we want $\DP[C_1,C_2,q]$ to store $-\infty$.
\smallskip

As a basic case,
given $q\in [\var_{\ex}]$ and $q$-grounding sets of colors $C_1$ and $C_2$.
If $\Hbar{p} \le 0$ for each $p \in [q-1]$, we store $\DP[C_1,C_2,q] = 0$.
Otherwise, if $C_1$ and $C_2$ are $q$-grounding but $\Hbar{p} > 0$ for any $p \in [q-1]$, we store $\DP[C_1,C_2,q] = -\infty$.

To compute further values, we use the recurrence
\begin{equation}
	\label{eqn:split-Dbar}
	\DP[C_1,C_2,q] =
	\max_{(x,v,e)} \DP[C_1',C_2',\ex^*(x)] + \ell(x).
\end{equation}
Here $C_1' := C_1\setminus c(\Pvx)$ and $C_2' := (C_2\cup c(\Pvx)) \setminus \{\hat c(e)\}$.
The maximum is taken over $(C_1,C_2)$-good tuples $(x,v,e)$ for which
$x\in Z_q$ and $\DP[C_1',C_2',\ex^*(x)] + \ell(x) \ge \max\{\Hbar{\ex^*(x)},\Hbar{\ex^*(x)+1},\dots,\Hbar{q-1}\}$ holds.
If no such tuple exists, we set $\DP[C_1,C_2,q] = -\infty$.
Recall that $\ex^*(x) = q$ for each $x\in Y_q$.
Observe that as~$C_1 \cap C_2 = \emptyset$, we also have $C_1' \cap C_2' = \emptyset$ by construction.

We return \yes if $\DP[C_1,C_2,\var_{\ex}] \ge \Hbar{\var_{\ex}}$ for some disjoint sets of colors~$C_1,C_2\subseteq [2\Dbar]$ with $|C_1| \leq \Dbar$.
Otherwise, if $\DP[C_1,C_2,\var_{\ex}] \ge \Hbar{\var_{\ex}}$ for all disjoint sets of colors~$C_1$ and $C_2$, we return \no.

\smallskip
We continue to analyze the algorithm by showing that $\DP[C_1,C_2,q]$ stores the largest rescue length $\ell(X(\mcA))$ of an ex-$(C_1,C_2,q)$-feasible anchored taxa set~$\mcA$.
We start with the basic case.
\begin{lemma}
	\label{lem:basic-case}
	$\DP[C_1,C_2,q]$ stores
	the largest rescue length $\ell(X(\mcA))$ of an ex-$(C_1,C_2,q)$-feasible anchored taxa set $\mcA$ if such an $\mcA$ exists, and $-\infty$ otherwise
	which is
	the desired value
	for each $q\in [\var_{\ex}]$ and all $q$-grounding sets of colors~$C_1$ and $C_2$.
\end{lemma}
\begin{proof}
	Let $q\in [\var_{\ex}]$ and 
	let $q$-grounding sets of colors $C_1$ and $C_2$ be given,
	and suppose for a contradiction that some non-empty anchored taxa set $\mcA$ is ex-$(C_1,C_2,q)$-feasible.
	Let $(x,v,e)$ be the last tuple in a valid ordering of $\mcA$.
	Such an ordering exists because $\mcA$ is color-respectful.
	Observe that $\hat c(e) \not\in c(E^+(\mcA))$, and in particular  $\hat c(e) \not\in c(\Pvx)$.
	As $\mcA$ holds CR a), CR c) and F d), we also have that $c(\Pvx)$ is uniquely colored, $\Pvx \subseteq E_{\le \Dbar}$ and $x\in Z_q$. Thus $(x,v,e) \in \mathcal X_{(q)}$.
	As $\mcA$ holds F b) we have $c(\Pvx) \subseteq C_1$.
	As $\mcA$ holds F c) we have $\hat C_2 \cup c(e) \in c(E^+(\mcA))$, which together with  $\hat c(e) \not\in c(E^+(\mcA))$ implies $\hat c(e) \in C_2$.
	But then $(x,v,e)$ is a tuple in $\mathcal X_{(q)}$ with $c(\Pvx) \subseteq C_1$ and $\hat c(e) \in C_2$, a contradiction to the assumption that $C_1$ and $C_2$ are $q$-grounding.	 
	It follows that no non-empty anchored taxa set is ex-$(C_1,C_2,q)$-feasible.
	
	Observe that the empty set is ex-$(C_1,C_2,q)$-feasible if and only if
	the empty set is ex-$q$-compatible, as the empty set trivially holds all other conditions.
	Since $\ell(\emptyset\cap Z_p)=\ell(\emptyset)=0$ for each $p\in [\var_{\ex}]$,
	we observe that the empty set is ex-$q$-compatible if and only if $\Hbar{p} \le 0$ for each $p\in [q-1]$.
	Exactly in these cases $\DP[C_1,C_2,q] \ge \ell(\emptyset) = 0$.
\end{proof}

In the following we show an induction over $|C_1|$.
Observe that $C_1=\emptyset$ together with any set of colors $C_2\subseteq [2\Dbar]$ are $q$-grounding for each $q \in [\var_{\ex}]$.
Therefore \Cref{lem:basic-case} is the basic case of the induction.

As induction hypothesis, we assume that for a fixed set of colors $C_1\subseteq [2\Dbar]$ and a fixed $q\in [\var_{\ex}]$, for each $K_1\subsetneq C_1$, $K_2 \subseteq [2\Dbar]\setminus K_1$ and $p \in [q]$ the entry $\DP[K_1,K_2,p]$ stores the largest rescue length $\ell(X(\mcA))$ of an ex-$(K_1,K_2,p)$-feasible anchored taxa set $\mcA$.

In \Cref{lem:corr-recurr-A} and \Cref{lem:corr-recurr-B} we proceed to show that with this hypothesis we can conclude that $\DP[C_1,C_2,p]$ stores the desired value.

\begin{lemma}
	\label{lem:corr-recurr-A}
	If an anchored taxa set $\mcA$ is ex-$(C_1,C_2,q)$-feasible for disjoint sets of colors~$C_1$ and $C_2$ that are not $q$-grounding, then $\DP[C_1,C_2,q]\ge \ell(\mcA)$.
\end{lemma}
\begin{proof}
	Let $C_1$ and $C_2$ be disjoint sets of colors that are not $q$-grounding.
	Furthermore, let $\mcA$ be an ex-$(C_1,C_2,q)$-feasible anchored taxa set.
	
	We prove the claim by induction on $|\mcA|$.
	First consider $\mcA$ being empty.
	Since $C_1$ and $C_2$ are not $q$-grounding, there is a tuple $(x,v,e)\in \mathcal X_{(q)}$ such that $\Pvx \subseteq C_1$ and $\hat c(e) \subseteq C_2$.
	As $(x,v,e) \in \mathcal X_{(q)}$, the anchored taxa set  $\mcA' := \{(x,v,e)\}$ is color-respectful, and by construction $\mcA$ holds F b), F c) and F d).
	It only remains to show that $\{x\}$ is ex-$q$-compatible to prove that $\mcA'$ is ex-$(C_1,C_2,q)$-feasible.
	Since $\mcA = \emptyset$ is ex-$(C_1,C_2,q)$-feasible, $\ell(x) > \ell(\emptyset) \ge \Hbar{p}$ for each $p\in [q-1]$.
	Thus, $\mcA'$ is ex-$(C_1,C_2,q)$-feasible.
	This is sufficient to show that $\DP[C_1,C_2,q] \ge \ell(x)$.
	So, we may assume that $\mcA$ is not empty.
	
	Now, let $(x,v,e)$ be the last tuple in valid the ordering of $\mcA$.
	Such an ordering exists because $\mcA$ is color-respectful.
	Define an anchored taxa set $\mcA'$ resulting from removing $(x,v,e)$ from \mcA.
	
	\medskip
	We first show that $\mcA'$ is ex-$(C_1',C_2',\ex^*(x))$-feasible.
	\proofpart{\hyperref[it:Fa]{F a) [$\mcA$ is color-respectful]}}
	Because $\mcA$ is color-respectful and $x$ is the taxon of largest order, we directly conclude that $\mcA'$ holds all properties of color-respectfulness.
	\proofpart{\hyperref[it:Fb]{F b) [$c(E^+(\mcA)) \subseteq C_1$]}}
	Because $E^+(\mcA') = E^+(\mcA) \setminus \Pvx$ and $E^+(\mcA)$ has unique colors, we conclude $c(E^+(\mcA')) = c(E^+(\mcA)) \setminus c(\Pvx) \subseteq C_1 \setminus c(\Pvx)= C_1'$.
	\proofpart{\hyperref[it:Fc]{F c) [$\hat c(e) \in C_2 \cup c(E^+(\mcA \setminus \{(x,v,e)\}))$ for each $(x,v,e)\in \mcA$]}}
	Fix a tuple $(y,u,e')\in \mcA'$.
	We want to show that $\hat c(e') \in$
	\begin{eqnarray*}
		&& C_2' \cup c(E^+(\mcA'\setminus \{(y,u,e')\}))\\
		&=& ((C_2 \cup c(\Pvx)) \setminus \{\hat c(e)\}) \cup (c(E^+(\mcA\setminus \{(y,u,e')\})) \setminus c(\Pvx))\\
		&=& (C_2 \cup c(E^+(\mcA\setminus \{(y,u,e')\})) \setminus \{\hat c(e)\}
	\end{eqnarray*}
	It holds $\hat c(e') \in C_2 \cup c(E^+(\mcA\setminus \{(y,u,e')\}))$, because $\mcA$ holds F c).
	Since $E_s(\mcA)$ has unique key-colors, $\hat c(e') \ne \hat c(e)$.
	Therefore $\hat c(e') \in C_2' \cup c(E^+(\mcA'\setminus \{(y,u,e')\}))$.
	\proofpart{\hyperref[it:Fd]{F d) [$X(\mcA) \subseteq Z_q$]}}
	$X(\mcA')\subseteq X(\mcA)\subseteq Z_{\ex^*(x)}$ by the definition of $x$.
	\proofpart{\hyperref[it:Fe]{F e) [$X(\mcA)$ is ex-$q$-compatible]}}
	Because $X(\mcA)$ is ex-$q$-compatible, $\ell(X(\mcA') \cap Z_p) = \ell(X(\mcA) \cap Z_p) \ge \Hbar{p}$ for each $p < \ex^*(x)$.
	Subsequently, $X(\mcA')$ is ex-$\ex^*(x)$-compatible.

	\medskip
	Next we show that $(x,v,e)$ is $(C_1,C_2)$-good.
	\proofpart{\hyperref[it:Ga]{G a) [\Pvx has unique colors]}}
	Because $E^+(\mcA)$ has unique colors, $\Pvx \subseteq E^+(\mcA)$ has unique colors, too.
	\proofpart{\hyperref[it:Gb]{G b) [$c(\Pvx) \subseteq C_1$]}}
	$c(\Pvx) \subseteq c(E^+(\mcA)) \subseteq C_1$ because $\mcA$ holds F b).
	\proofpart{\hyperref[it:Gc]{G c) [$\hat{c}(e) \in C_2$]}}
	Because $\mcA$ holds F c) and $X(\mcA)$ has a valid ordering, we conclude that
	$\hat c(e) \in C_2 \cup c(E^+(\mcA'))$ and
	$\hat c(e) \not\in c(E^+(\mcA))$.
	Therefore, especially $\hat c(e) \not\in c(E^+(\mcA')) \subseteq c(E^+(\mcA))$ and we conclude that
	$\hat c(e) \in C_2$.
	\proofpart{\hyperref[it:Gd]{G d) [$\Pvx$ and $E_{> \Dbar}$ are disjoint]}}
	Because $E^+(\mcA)$ has an empty intersection with $E_{> \Dbar}$, also $\Pvx \subseteq E^+(\mcA)$ has empty intersection with $E_{> \Dbar}$.
	
	\medskip
	Since $\mcA'$ is ex-$(C_1',C_2',\ex^*(x))$-feasible and $|\mcA'| < |\mcA|$, we have by the inductive hypothesis that $\DP[C_1',C_2',\ex^*(x)] \ge \ell(X(\mcA'))$.
	Furthermore $\DP[C_1',C_2',\ex^*(x)] + \ell(x) \ge \ell(X(\mcA')) + \ell(x) \ge \ell(X(\mcA))$. As $X(\mcA)$ is ex-$(C_1,C_2,q)$-compatible we have  $ \ell(X(\mcA) = \ell(X(\mcA)\cap Z_p)  \geq \Hbar{p}$ for each $p$ with $\ex^*(x) \le p < q$, and so $\DP[C_1',C_2',\ex^*(x)] + \ell(x)  \ge \max\{\Hbar{\ex^*(x)},\Hbar{\ex^*(x)+1},\dots,\Hbar{q-1}\}$.
	This together with the fact that $(x,v,e)$ is $(C_1,C_2)$-good implies that $(x,v)$ satisfies the conditions of Recurrence~(\ref{eqn:split-Dbar}), from which we can conclude that $\DP[C_1,C_2,q] \geq \DP[C_1',C_2',\ex^*(x)] + \ell(x) \geq \ell(X(\mcA))$.
\end{proof}

\begin{lemma}
	\label{lem:corr-recurr-B}
	If $\DP[C_1,C_2,q] = a \ge 0$ for some disjoint sets of colors $C_1$ and $C_2$, then there is an ex-$(C_1,C_2,q)$-feasible anchored taxa set $\mcA$ with $\ell(X(\mcA)) = a$.
\end{lemma}
\begin{proof}
	Assume that $\DP[C_1,C_2,q] = 0$.
	Then $C_1$, and $C_2$ must be $q$-grounding, as otherwise the algorithm would apply Recurrence (\ref{eqn:split-Dbar}) and then~$\DP[C_1,C_2,q]$ would store either $-\infty$ or a value which is at least $\ell(x) > 0$ for some $x\in X$.
	\Cref{lem:basic-case} shows the correctness of this case.
	
	Now assume that $\DP[C_1,C_2,q] = a > 0$.
	By Recurrence (\ref{eqn:split-Dbar}) there is a $(C_1,C_2)$-good tuple $(x,v,e)$ such that $x\in Z_q$ and $\DP[C_1,C_2,q] = \DP' + \ell(x)$.
	Here, $\DP' := \DP[C_1',C_2',q']$ where (as in Recurrence (\ref{eqn:split-Dbar})) $C_1' := C_1\setminus c(\Pvx)$, $C_2' := (C_2\cup c(\Pvx)) \setminus \{\hat c(e)\}$, and $q' := \ex^*(x)$.
	Further we know that $\DP' + \ell(x) \ge \max\{\Hbar{q'},\Hbar{q'+1},\dots,\Hbar{q-1}\}$.
	We conclude by the induction hypothesis that there is an ex-$(C_1',C_2',q')$-feasible anchored taxa set~$\mcA'$ such that $\ell(X(\mcA')) = \DP'$.
	Define $\mcA := \mcA' \cup \{(x,v,e)\}$.
	Clearly, $a = \DP' + \ell(x) = \ell(X(\mcA))$.
	It remains to show that $\mcA$ is an ex-$(C_1,C_2,q)$-feasible set.
	
	\medskip
	\proofpart{\hyperref[it:Fa]{F a) [$\mcA$ is color-respectful]}}
	We show that $\mcA$ is color-respectful.
	
	\proofpart{\hyperref[it:Ca]{CR a) [$E^+(\mcA)$ has unique colors]}}
	As $\mcA'$ holds CR a) and F b), $E^+(\mcA')$ has unique colors and $c(E^+(\mcA')) \subseteq C_1' = C_1\setminus c(\Pvx)$.
	Subsequently, the colors of $E^+(\mcA')$ and $\Pvx$ are disjoint.
	We conclude with the $(C_1,C_2)$-goodness of $(x,v,e)$, that \Pvx has unique colors and so also $E^+(\mcA)$.
	\proofpart{\hyperref[it:Cb]{CR b) [$E_s(\mcA)$ has unique key-colors]}}
	As $\mcA'$ holds CR b), $E_s(\mcA')$ has unique key-colors.
	Observe $E_s(\mcA) = E_s(\mcA') \cup \{e\}$.
	Because $\mcA'$ holds F c), $\hat c(E_s(\mcA')) \subseteq C_2' \cup c(E^+(\mcA')) \subseteq C_1' \cup C_2'$.
	On the other hand, because $\hat c(e) \in C_2$ and $C_1\cap C_2 = \emptyset$ we have  $\hat c(e) \notin C_1$, which implies
	 $\hat c(e) \not\in C_1\cup (C_2 \setminus \{\hat c(e)\}) = C_1'\cup C_2'$. 
	Therefore, $E_s(\mcA)$ has unique key-colors.
	\proofpart{\hyperref[it:Cc]{CR c) [$E^+(\mcA)$ and $E_{> \Dbar}$ are disjoint]}}
	$E^+(\mcA')$ and $E_{>\Dbar}$ are disjoint because $\mcA'$ is color-respectful.
	Further because $(x,v,e)$ is $(C_1,C_2)$-good, $\Pvx$ and $E_{>\Dbar}$ are disjoint and so $E^+(\mcA') \subseteq E_{\le \Dbar}$.
	\proofpart{\hyperref[it:Cd]{CR d) [$\mcA$ has a valid ordering]}}
	$\mcA'$ has a valid ordering $(y_1,u_1,e_1),\dots,(y_{|\mcA'|},u_{|\mcA'|},e_{|\mcA'|})$.
	As $X(\mcA')\subseteq Z_{\ex^*(x)}$ we conclude $\ex(y) \le \ex(x)$ for each $y\in A'$.
	Further $\hat c(e) \in C_2$ and $c(\Pvx) \subseteq C_1$ because $(x,v,e)$ is $(C_1,C_2)$-good.
	Since $\mcA'$ holds F b), $c(E^+(A')) \subseteq C_1' \subseteq C_1$.
	Because $C_1$ and $C_2$ are disjoint, we conclude that $\hat c(e) \not \in c(E^+(\mcA))$.
	Hence, $(y_1,u_1,e_1),\dots,(y_{|\mcA'|},u_{|\mcA'|},e_{|\mcA'|}),(x,v,e)$ is a valid ordering of $\mcA$.
	\proofpart{\hyperref[it:Ce]{CR e) [$\Pvx$ and $P_{u,y}$ are disjoint for any $(x,v,e),(y,u,e') \in \mcA$]}}
	We know that $\mcA'$ holds CR e).
	Further, as already seen, $c(E^+(\mcA'))$ and $c(\Pvx)$ are disjoint and therefore also \Pvx and $P_{u,y}$ are disjoint for each $(y,u,e') \in \mcA'$.\linebreak 
	It follows that $\mcA$ is color-respectful.
	
	\medskip
	\proofpart{\hyperref[it:Fb]{F b) [$c(E^+(\mcA)) \subseteq C_1$]}}
	$c(E^+(\mcA)) = c(E^+(\mcA')) \cup c(\Pvx) \subseteq C_1' \cup c(\Pvx) = C_1$.
	\proofpart{\hyperref[it:Fc]{F c) [$\hat c(e) \in C_2 \cup c(E^+(\mcA \setminus \{(x,v,e)\}))$ for each $(x,v,e)\in \mcA$]}}
	Observe $\hat c(e) \in C_2$ because $(x,v,e)$ is $(C_1,C_2)$-good.
	Each $(y,u,e') \in \mcA'$ holds $\hat c(e') \subseteq C_2' \cup c(E^+(\mcA'\setminus \{(y,u,e')\}))$, as $\mcA'$ holds F c).
	Observe $x\ne y$ and $C_2' \subseteq C_2 \cup c(\Pvx)$ and so $\hat c(e') \subseteq C_2 \cup c(E^+(\mcA\setminus \{(y,u,e')\}))$.
	\proofpart{\hyperref[it:Fd]{F d) [$X(\mcA) \subseteq Z_q$]}}
	$X(\mcA)\subseteq Z_q$ because $X(\mcA')\subseteq Z_q$ and $x\in Z_q$.
	\proofpart{\hyperref[it:Fe]{F e) [$X(\mcA)$ is ex-$q$-compatible]}}
	Because $\mcA'$ is ex-$q'$-compatible, we conclude that $\ell(X(\mcA)\cap Z_p) = \ell(X(\mcA')\cap Z_p) \ge \Hbar{p}$ for each $p\in [q'-1]$.
	Then with $\DP' + \ell(x) \ge \max\{\Hbar{q'},\Hbar{q'+1},\dots,\Hbar{q-1}\}$ we conclude that $\ell(X(\mcA)\cap Z_p) = \ell(X(\mcA)) \ge \Hbar{p}$ for each $p\in \{q',\dots,q-1\}$.
	Therefore, $X(\mcA)$ is ex-$q$-compatible and ex-$(C_1,C_2,q)$-feasible.
\end{proof}

\begin{lemma}

	An instance  \Instance of \cBartPDs\Dbar is a \yes-instance if and only if there exist disjoint sets of colors $C_1,C_2 \subseteq [2\Dbar]$ with $|C_1| \leq \Dbar$ and an  ex-$(C_1,C_2,\var_{\ex})$-feasible anchored taxa set $\mcA$ with $\ell(X(\mcA)) \ge \Hbar{\var_{\ex}}$.
\end{lemma}
\begin{proof}
	Suppose first that \Instance  is a \yes-instance of \cBartPDs\Dbar. That is, there exists a color-respectful anchored taxa set $\mcA$ with $|c(E^+(\mcA))| \leq \Dbar$ such that there is a valid $T$-schedule saving $S := X\setminus X(\mcA)$.
	Let $C_1 := c(E^+(\mcA))$
	and let $C_2$ be the key-colors of the edges $E_s(\mcA)$ that are not in $C_1$.
	We show that $\mcA$ is ex-$(C_1,C_2,\var_{\ex})$-feasible.
	First, $\mcA$ holds F a), F b), and F d) by definition.
	\proofpart{\hyperref[it:Fc]{F c) [$\hat c(e) \in C_2 \cup c(E^+(\mcA \setminus \{(x,v,e)\}))$ for each $(x,v,e)\in \mcA$]}}
	By the definition, $\hat c( E(\mcA)) \subseteq C_2 \cup C_1 = C_2 \cup c(E^+(\mcA))$.
	Fix a tuple $(x,v,e) \in \mcA$.
	Because $\mcA$ has a valid ordering, we conclude that $\hat c(e)\not\in c(\Pvx)$.
	We conclude that if $\hat c(e)$ is in $c(E^+(\mcA)))$, then explicitly $\hat c(e) \in c(E^+(\mcA\setminus \{(x,v,e)\}))$.
	Therefore $\mcA$ holds F c).
	\proofpart{\hyperref[it:Fe]{F e) [$X(\mcA)$ is ex-$q$-compatible]}}
	Fix some $q\in [\var_{\ex}]$.
	Since $S$ has a valid schedule, we conclude $\ell(S\cap Z_q) \le H_q$ for each $q\in [\var_{\ex}]$.
	Subsequently $\ell(X(\mcA)\cap Z_q) = \ell(Z_q) - \ell(S\cap Z_q) \ge \Hbar{q}$.
	Likewise with $\ell(S) \le H_{\var_{\ex}}$ we conclude that $\ell(X(\mcA)) \ge \Hbar{\var_{\ex}}$.
	
	\medskip
	For the converse, suppose $\mcA$ is an ex-$(C_1,C_2,\var_{\ex})$-feasible anchored taxa set with $\ell(X(\mcA)) \ge \Hbar{\var_{\ex}}$ for disjoint sets of colors $C_1,C_2 \subseteq [2\Dbar]$ with $|C_1|\leq \Dbar$.
	Fix some $q\in [\var_{\ex}]$.
	Observe that the rescue time of all taxa in $S=X\setminus X(\mcA)$ that are in $Z_q$ is $\ell(S \cap Z_q) = \ell((X\setminus X(\mcA)) \cap Z_q) = \ell(Z_q) - \ell(X(\mcA)\cap Z_q)$.
	Because $\mcA$ is ex-$\var_{\ex}$-compatible, $\ell(X(\mcA)\cap Z_q) \ge \Hbar{q}$ and so
	$\ell(S \cap Z_q) \le \ell(Z_q) - \Hbar{q} = H_q$. 
	Then by \Cref{lem:scheduleCondition}, there is a valid $T$-schedule saving $S$.
	By definition, $\mcA$ is color-respectful
	and $|c(E^+(\mcA))| \leq |C_1| \leq \Dbar$.
	Therefore \Instance  is a \yes-instance of \cBartPDs\Dbar
\end{proof}

\begin{lemma}
	\label{lem:rt-DBar}
	Algorithm~$(\Dbar)$ computes a solution of \cBartPDs\Dbar in time $\Oh^*(9^\Dbar \cdot \Dbar)$.
\end{lemma}
\begin{proof}
	Since $\DP[C_1,C_2,q]$ is computed for $C_1$ and $C_2$ disjoint subsets of $[2\Dbar]$,
	the table $\DP$ contains $3^{2\Dbar} \cdot \var_{\ex}$ entries.
	
	For given $q\in [\var_{\ex}]$ and color sets $C_1$ and $C_2$, we can compute whether $C_1$ and $C_2$ are $q$-grounding by testing the conditions for each tuple $(x,v,e)$.
	Therefore the test can be done in $\Oh(\Dbar \cdot n^3)$ time.
	Here the $\Dbar$-factor in the running time comes from testing whether $\Pvx \subseteq C_1$.
	
	Analogously, for given $q\in [\var_{\ex}]$ and color sets $C_1$ and $C_2$, we can compute whether a tuple $(x,v,e)$ holds the conditions of Recurrence~(\ref{eqn:split-Dbar}) in 
	$\Oh(\Dbar \cdot n \cdot \log \Hbar{*})$ time,
	where $\Hbar{*}:= \max\{\Hbar{1}, \Hbar{2}, \dots, \Hbar{\var_{\ex}}\}$.
	The $\log \Hbar{*}$-factor comes from adding two numbers $\DP[C_1',C_2',\ex^*(x)]$ and $\ell(x)$ 
	and comparing the given sum with $\max\{\Hbar{\ex^*(x)},\Hbar{\ex^*(x)+1},\dots,\Hbar{q-1}\}$.
\end{proof}

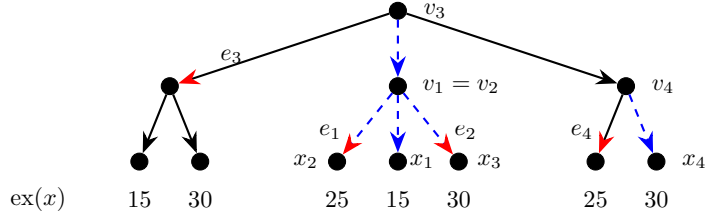
\begin{figure}[t]
	\centering
	\begin{tikzpicture}[scale=1,every node/.style={scale=0.9}]
		\tikzstyle{txt}=[circle,fill=white,draw=white,inner sep=0pt]
		\tikzstyle{nde}=[circle,fill=black,draw=black,inner sep=2.5pt]
		
		\node[nde] (v0) at (5,10) {};
		\node[nde] (v1) at (2,9) {};
		\node[nde] (v2) at (5,9) {};
		\node[nde] (v3) at (8,9) {};
		\node[nde] (v4) at (1.6,8) {};
		\node[nde] (v5) at (2.4,8) {};
		\node[nde] (v6) at (4.2,8) {};
		\node[nde] (v7) at (5,8) {};
		\node[nde] (v7') at (5.8,8) {};
		\node[nde] (v8) at (7.6,8) {};
		\node[nde] (v9) at (8.4,8) {};

		\node[txt,xshift=9mm,yshift=-10mm] (c1) [above=of v1] {$e_3$};

		\node[txt,xshift=-1mm,yshift=-10mm] (c6) [above=of v6] {$e_1$};
		\node[txt,xshift=1mm,yshift=-10mm] (c8) [above=of v7'] {$e_2$};
		\node[txt,xshift=-2mm,yshift=-10mm] (c8) [above=of v8] {$e_4$};
		
		\node[txt,yshift=9mm] (r4) [below=of v4] {$15$};
		\node[txt,yshift=9mm] (r5) [below=of v5] {$30$};
		\node[txt,yshift=9mm] (r6) [below=of v6] {$25$};
		\node[txt,yshift=9mm] (r7) [below=of v7] {$15$};
		\node[txt,yshift=9mm] (r7) [below=of v7'] {$30$};
		\node[txt,yshift=9mm] (r8) [below=of v8] {$25$};
		\node[txt,yshift=9mm] (r9) [below=of v9] {$30$};
		\node[txt,xshift=3mm] (r'4) [left=of r4] {$\ex(x)$};
 		
		\node[txt,xshift=-9mm] (t0) [right=of v0] {$v_3$};
		\node[txt,xshift=9mm] (t1) [left=of v1] {};
		\node[txt,xshift=-9mm] (t2) [right=of v2] {$v_1=v_2$};
		\node[txt,xshift=-9mm] (t3) [right=of v3] {$v_4$};
		
		\node[txt,xshift=10mm] (t6) [left=of v6] {$x_2$};
		\node[txt,xshift=-11mm] (t7) [right=of v7] {$x_1$};
		\node[txt,xshift=-10mm] (t7) [right=of v7'] {$x_3$};
		\node[txt,xshift=-9mm] (t9) [right=of v9] {$x_4$};
		\textbf{}
		\draw[thick,arrows = {-Stealth[red,length=8pt]}] (v0) to (v1);
		\draw[thick,dashed,blue,arrows = {-Stealth[length=8pt]}] (v0) to (v2);
		\draw[thick,arrows = {-Stealth[length=8pt]}] (v0) to (v3);
		\draw[thick,arrows = {-Stealth[length=8pt]}] (v1) to (v4);
		\draw[thick,arrows = {-Stealth[length=8pt]}] (v1) to (v5);
		\draw[thick,dashed,blue,arrows = {-Stealth[red,length=8pt]}] (v2) to (v6);
		\draw[thick,dashed,blue,arrows = {-Stealth[length=8pt]}] (v2) to (v7);
		\draw[thick,dashed,blue,arrows = {-Stealth[red,length=8pt]}] (v2) to (v7');
		\draw[thick,arrows = {-Stealth[red,length=8pt]}] (v3) to (v8);
		\draw[thick,dashed,blue,arrows = {-Stealth[length=8pt]}] (v3) to (v9);
	\end{tikzpicture}
	\caption{
		Example of the construction of an anchored taxa set $\mcA$ (\Cref{thm:DBar}).
		Here given $X\setminus S = \{x_1,x_2,x_3,x_4\}$, the constructed $\mathcal A$ is $\{ (x_i,v_i,e_i) \mid i \in [4] \}$.
		Blue dashed edges are in $E_d(X\setminus S) = E^+(\mcA)$ and edges with red arrow are in $E_s(\mcA)$.
		$\mcA_1$ could have been $\{ (x_1,v_1,e_2) \}$ provoking that $(x_1,v_1,e_2)$ must have been replaced in step $i=2$.
	}
	\label{fig:anchored-taxa-set}
\end{figure}%

\begin{theorem}
	\label{thm:DBar}
	\tPDs can be solved in
	$\Oh^*(2^{6.056 \cdot \Dbar + o(\Dbar)})$~time.
\end{theorem}

Just like in the previous section,
the key idea is that we construct a family of colorings $\mathcal{C}$ on the edges of $\Tree$, where each edge $e\in E$ is assigned a key-color and additionally for each $e\in E_{\le \Dbar}$ a subset $c^-(e)$ of $[2\Dbar]$ of size $\w(e)-1$.
Using these, we generate~$|\mathcal{C}|$ instances of \cBartPDs\Dbar, which with Algorithm~$(\Dbar)$ we solve in $\Oh^*(9^\Dbar \cdot \Dbar)$ time.
The colorings are constructed in such a manner that $\Instance$ is a \yes-instance if and only if at least one of the constructed \cBartPDs\Dbar instances is a \yes-instance.
Central for this will be the concept of a perfect hash family, see \Cref{def:hash-fam}.

\begin{proofm}{of \Cref{thm:DBar}}
	\proofpara{Reduction}
	Let $\Instance = (\Tree, r, \ell, T, D)$ be an instance of \tPDs.
	Let $|E_{\le \Dbar}| = m_1$ and $|E_{> \Dbar}| = m - m_1$.
	Order the edges $e_1, \dots, e_{m_1}$ of $E_{\le \Dbar}$
	and the edges $e_{m_1+1}, \dots, e_{m}$ of $E_{> \Dbar}$ arbitrarily.
	We define $W_0 := m$ and $W_j := m + \sum_{i=1}^{j} (\w(e_{i}) - 1)$ for each $j\in [m_1]$.
	
	Let $\mathcal{F}$ be a $(W_m, 2\Dbar)$-perfect hash family.
	Now we define the family of colorings $\mathcal{C}$ as follows.
	For every $f \in \mathcal{F}$,
	let $\hat c_f$ and $c^-_f$ be colorings such that
	$\hat c_f$ assigns edge $e_j \in E$ the color $f(j)$ for each $j \in [m]$
	and $c^-_f$ assigns edge $e_j \in E_{\le \Dbar}$ the colors $\{f(W_{j-1}+1), \dots, f(W_j)\}$ for each $j\in [m_1]$.
	
	For each $c_f \in \mathcal{C}$,
	let $\Instance_{f} = (\Tree, r, \ell, T, D, \hat c_f, c^-_f)$ be the corresponding instance of \cBartPDs\Dbar.
	Now solve each $\Instance_{f}$ using Algorithm~$(\Dbar)$ and return \yes if and only if $\Instance_{f}$ is a \yes-instance for some $c_f \in \mathcal{C}$.
	
	\proofpara{Correctness}
	If one of the  constructed instances of \cBartPDs\Dbar is a \yes-instance, then $\Instance$ is a \yes-instance of \tPDs by \Cref{lem:coloredYes}.
	
	For the converse, 
	if $S\subseteq X$ is a solution for $\Instance$, then $\w(E') \ge D$ where $E'$ is the set of edges which have at least one offspring in $S$.
	Therefore, $\w(E\setminus E') \le \w(E) - D = \Dbar$.
	Define $E_1 := E\setminus E' = E_d(X\setminus S)$ and let $E_2'$ be the set of edges $e\in E'$ which have a sibling-edge in $E_1$.
	If $e,e' \in E_2'$ are sibling-edges then remove $e'$ from $E_2'$.
	Continuously repeat the previous step to receive $E_2$ in which any two edges $e,e' \in E_2$ are not sibling-edges.
	Observe $|E_2| \le |E_1| \le \w(E_1) \le \Dbar$.
	
	Let $Z \subseteq [W_m]$ be the subset of indices corresponding to the colors of $E_1$ and the key-colors of $E_2$. 
	More precisely $Z:= \{ j\in [m] \mid e_j \in E_1 \cup E_2\} \cup \{W_{j-1} +1, \dots, W_j \mid e_j \in E_1\}$.
	Since $\mathcal{F}$ is a $(W_m, 2\Dbar)$-perfect hash family and $|Z| = \w(E_1) + |E_2| \leq 2\Dbar$, there exists a function $f \in \mathcal{F}$ such that $f(z) \neq f(z')$ for distinct $z, z' \in Z$.
	It follows that $E_1$ has unique colors, $E_2$ has unique key-colors, and $c(E_1)\cap \hat c(E_2) = \emptyset$ in all constructed instance $\Instance_f$.

	Let $C_1  = c(E_1)$ and $C_2 = \hat c(E_2)$.
	We claim that $\Instance_f$ has a color-respectful anchored taxa set $\mcA$ with $E^+(\mcA) = E_1$, $E_s(\mcA) \subseteq E_1 \cup E_2$ and $X(\mcA) = X\setminus S$. Since $|c(E_1)| = |C_1| \leq \Dbar$ and there is a valid $T$-schedule saving $S$, this implies that  $\Instance_f$ is a \yes-instance of \cBartPDs\Dbar.
	
		We define $\mcA := (\{(x_i,v_i,e_i\}: 1 \leq i \leq |X \setminus S|)$, where $x_i,v_i,e_i$ are constructed iteratively as follows.
		Let $x_1,\dots, x_{|X\setminus S|}$ be an ordering of $X\setminus S$ such that $\ex(i) \leq \ex(j)$ if $i \leq j$.
		Define $E_d^{(i)} := E_d(\{x_1,\dots, x_i\})$ for each  $1 \leq i \leq |X \setminus S|$ and let $E_d^0 := \emptyset$.
		For each $1 \leq i \leq |X \setminus S|$,
		let $e' := v_iw_i$ be the highest edge in $E_d^{(i)} \setminus E_d^{(i-1)}$.
		This completes the construction of $x_i, v_i$. Construct $e_i$ as follows.
		If $v_j \neq v_i$ for each $j > i$, then let $e_i$ be an arbitrary outgoing edge of $v_i$ not in $E_d^{(i)}$.
		Such an edge exists as $v_iw_i$ is the topmost edge of $E_d^{(i)} \setminus E_d^{(i-1)}$ and therefore not all edges outgoing of $v_i$ are in $E_d^{(i)}$.
		Otherwise, let $j_i$ be the minimum integer such that $j_i>i$ and $v_{j_i}=v_i$. Then let $e_i :=v_iw_{j_i}$. Note that $e_i$ is not in $E_d^{(i)}$ since $e_i$ is the topmost edge of $E_d^{(j_i)} \setminus E_d^{(j_i-1)}$.
		See \Cref{fig:anchored-taxa-set} for an example.
	
	It remains to show that $\mcA$ is color-respectful.
	$\mcA$ holds CR a) as by construction $E^+(\mcA) = \bigcup_{i=1}^{|X \setminus S|} P_{v_i,x_i} = E_d(X\setminus S) = E_1$.
	Similarly $\mcA$ holds CR b) as $E_s(\mcA) \subseteq E_1 \cup E_2$.
	$\mcA$ holds CR c) by the fact that $\w(E_1) \leq \Dbar$.
	$\mcA$ holds CR d), as $(x_1,v_1,e_1),\dots (x_{|\mcA|},v_{|\mcA|},e_{|\mcA|})$ is a valid ordering of $\mcA$.
	
	To see that $\mcA$ holds CR e), consider any $(x_i,v_i,e_i), (x_j,v_j,e_j)$ for $i < j$. 
	As every $v_j$ is an ancestor of $x_j$, every edge in $P_{v_j,x_j}$ has an offspring not in $\{x_1,\dots, x_i\}$ and therefore $P_{v_j,x_j}$ is disjoint from $E_d(x_1,\dots, x_i)$, which contains all edges of $P_{v_i,x_i}$. It follows that  $P_{v_i,x_i}$ and $P_{v_j,x_j}$ are disjoint.
	
	\proofpara{Running Time}
	The construction of $\mathcal{C}$ takes $e^{2\Dbar} (2\Dbar)^{\Oh(\log {(2\Dbar)})} \cdot W_m \log W_m$ time, and for each $c \in \mathcal{C}$ the construction of each instance of \cBartPDs\Dbar takes time polynomial in $|\Instance|$. By \Cref{lem:rt-DBar}, solving an instance of \cBartPDs\Dbar takes $\Oh^*(9^{\Dbar}\cdot {2\Dbar})$~time, and $|\mathcal{C}| = e^{2\Dbar} (2\Dbar)^{\Oh(\log {(2\Dbar)})} \cdot \log W_m$ is the number of constructed instances of \cBartPDs\Dbar.
	
	Thus,
	$\Oh^*(e^{2\Dbar} (2\Dbar)^{\Oh(\log {(2\Dbar)})} \log W_m \cdot (W_m + (9^{\Dbar}\cdot {2\Dbar})))$
	is the total running time.
	Because $W_m = \w(E_{\le \Dbar}) + |E_{>\Dbar}| \le 
	|E|\cdot\Dbar \le 2n \cdot \Dbar$
	this further simplifies
	to~$\Oh^*((3e)^{2\Dbar} \cdot 2^{\Oh(\log^2(\Dbar))}) = \Oh^*(2^{6.056 \cdot \Dbar + o(\Dbar)})$.
\end{proofm}

\section{Additional parameterized complexity results}
\label{sec:other}

\subsection{Parametrization by the available person-hours}
In this subsection we consider parametrization by the available person-hours~$H_{\max_{\ex}}$.
Observe we may assume that $|T|,\max_{\ex} \le H_{\max_{\ex}} \le |T| \cdot \max_{\ex}$.

\begin{proposition}
	\label{prop:T+maxr}
	\tPDs and \tPDws are \FPT when parameterized by the available person-hours $H_{\max_{\ex}}$.
	More precisely,
	\begin{enumerate}[(a)]
		\itemsep-.35em
		\item\label{prop:s-T^maxr} \tPDs can be solved in $\Oh((|T| + 1)^{2\cdot\max_{\ex}} \cdot n \cdot \log(D))$ time,
		\item\label{prop:s-T*maxr} \tPDs can be solved in $\Oh((H_{\max_{\ex}})^{2\var_{\ex}} \cdot n \cdot \log(D))$ time, and
		\item\label{prop:ws-T+maxr} \tPDws can be solved in $\Oh(3^{|T| \cdot \max_{\ex}} \cdot |T| \cdot n \cdot \log(D))$ time.
	\end{enumerate}
\end{proposition}
\begin{proofm}{of \Cref{prop:T+maxr}\ref{prop:s-T^maxr}}
	\proofpara{Table definition}
	Let \Instance be an instance of \tPDs.
	Let $\vec{a} = (a_1,\dots,a_{\max_{\ex}})$ be a $\max_{\ex}$-dimensional vector with entries $a_i$ in $[|T|]_0$.
	
	We define a dynamic programming algorithm with a table $\DP$ in which in entry $\DP[v,\vec{a},b]$ with $b\in \{0,1\}$ we store 0 if $b=0$;
		if $b = 1$, then we store the maximum value 
		of $\PDgiveTree{\Tree_v}(\{S\})$ for a non-empty subset of taxa $S$ below $v$
 that can be saved when having $a_j$ available teams in timeslot $j\in [\max_{\ex}]$ (that is, such that $\sum_{x \in S \cap Y_j} \leq a_j$ for all $j \in [\max_{\ex}]$).
	If there is no such non-empty $S$, we store $-\infty$.
	We define an auxiliary table $\DP_1$ in which in entry $\DP[v,\vec{a},b,i]$ we only consider the first $i$ children of~$v$.
	
	\proofpara{Algorithm}
	For each taxon $x$ we store $\DP[x,\vec{a},b] = 0$ if $\sum_{i=1}^{\ex(x)} a_i \ge \ell(x)$ or if~$b=0$.
	Otherwise, we store $\DP[x,\vec{a},1] = -\infty$.
	
	Let $v$ be an internal vertex with children $u_1,\dots,u_z$.
	We define $\DP_1[v,\vec{a},b,1] := \DP[u_1,\vec{a},b] + \w(v u_1) \cdot b$
	and we compute further values with the recurrence
	\begin{equation}
		\label{eqn:s-T^maxr}
		\DP_1[v,\vec{a},b,i+1] =
		\max_{\vec{d}, b_1,b_2}
		\DP_1[v,\vec{d},b_1,i] + \DP[u_{i+1},\vec{a}-\vec{d},b_2] + 
		\w(v u_i) \cdot b_2.
	\end{equation}
	Here, we select $\vec{d}$ such that $d_j\in [a_j]_0$ for each $j\in [\max_{\ex}]$ and $b_1,b_2$ are selected to be in $[b]_0$.
	We finally set $\DP[v,\vec{a},b] := \DP_1[v,\vec{a},b,z]$.

	We return \yes, if $\DP[\rho,\vec{a}^*,1] \ge D$ where $\rho$ is the root of the given phylogenetic tree \Tree and $a_i^*$ is the number of teams $t_j = (s_j,e_j)$ with $s_j< i\le e_j$. Otherwise, we return \no.

	\proofpara{Correctness}
	In vector $\vec{a}^*$ the number of available teams per timeslot is stored, such that \Instance is a \yes-instance if and only if $\DP[\rho,\vec{a}^*,1] \ge D$. It remains to show that $\DP$ stores the intended value.
	
	If $x$ is a leaf then the subtree $\Tree_x$ rooted at $x$ does not contain edges and so $\PDgiveTree{\Tree_x}(\{x\}) = \PDgiveTree{\Tree_x}(\emptyset) = 0$.
	As $S=\{x\}$ is the only non-empty set of taxa in $\Tree_x$, in the case that $b=1$, we need to ensure that in $\vec{a}$ enough person-hours are available for $\{x\}$, which is $\sum_{i=1}^{\ex(x)} a_i \ge \ell(x)$. Thus, the basic case is correct.
	
	Let $v$ be a vertex with children $u_1,\dots,u_z$.
	To see that $\DP_1[v,\vec{a},b,i]$ for $i\in [q]$ stores the correct value, observe that the diversity of edge $vu_i$ can be added if and only if at least one taxon in $\Tree_{u_i}$ survives, which is happening if and only if $b_2=1$ (or $b=1$ in the case of $i=1$).
	Further, the available person-hours at $v$ can naturally be divided between the children $v$.
	Therefore, $\DP_1$ stores the correct value.
	The correctness of $\DP[v,\vec{a},b]$ directly follows from the correctness of $\DP_1[v,\vec{a},b,q]$.

	\proofpara{Running time}
	The table $\DP$ contains $n \cdot (|T| + 1)^{\max_{\ex}} \cdot 2$ entries.
	For a leaf $x$ the value of $\DP[x,\vec{a},1]$ can be computed by computing $\sum_{i=1}^{\ex(x)} a_i$ in $\Oh(\log(H_{\max_{\ex}}))$ time.
	For an internal vertex $v$ with $q$ children, we need to copy the value of $\DP_1[v,\vec{a},b,q]$ in $\Oh(\log(D))$ time.
	
	The table $\DP_1$ contains $n \cdot (|T| + 1)^{\max_{\ex}} \cdot 2$ entries as well, as each vertex is child of exactly one internal vertex.
	In Recurrence~(\ref{eqn:s-T^maxr}), there are $\Oh((|T| + 1)^{\max_{\ex}})$ options to choose $\vec{d}$ and at most $4$ options to choose $b_1$ and $b_2$,
	and so the values of the entries can be computed in $\Oh((|T| + 1)^{\max_{\ex}} \cdot n \cdot \log(D))$ time.
	
	Finally, we need to iterate over the options for $\vec{a^*}$.
	Altogether we can compute a solution for an instance of \tPDs in $\Oh((|T| + 1)^{2\cdot\max_{\ex}} \cdot n \cdot \log(D))$ time.
	
	Observe that we do not need to name the additional $\log(H_{\max_{\ex}}) \le \log(|T|\cdot \max_{\ex})$ factor of the basic case, as $(|T| + 1)^{\max_{\ex}}$ is bigger.
\end{proofm}
\begin{proofm}{of \Cref{prop:T+maxr}\ref{prop:s-T*maxr}}
	\proofpara{Table definition}
	Let $\vec{a} = (a_1,\dots,a_{\var_{\ex}})$ be a $\var_{\ex}$-dimensional vector with values $a_i$ in $[H_i]_0$.
	
	We define a dynamic programming algorithm with a table $\DP$ in which in entry $\DP[v,\vec{a},b]$ with $b\in \{0,1\}$ we store 0 if $b=0$;
 	if $b = 1$, then we store the maximum value of $\PDgiveTree{\Tree_v}(\{S\})$ for a subset of taxa $S$ below $v$ such that $\sum_{x \in S \cap Z_j} \leq a_j$ for all $j \in [\var_{\ex}]$.
	We define an auxiliary table $\DP_1$ in which in entry $\DP[v,\vec{a},b,i]$ we only consider the first $i$ children of $v$.
	
	\proofpara{Algorithm}
	For each leaf $x$ we store $\DP[x,\vec{a},b] = 0$ if $b=0$ or $a_i \ge \ell(x)$ for each $i \ge \ex(x)$.
	Otherwise, we store $\DP[x,\vec{a},1] = -\infty$.
	
	Let $v$ be an internal vertex with children $u_1,\dots,u_z$.
	We define $\DP_1[v,\vec{a},b,1] := \DP[u_1,\vec{a},b] + \w(v u_1) \cdot b$
	and we compute further values with the recurrence
	\begin{equation}
		\label{eqn:T*maxr}
		\DP_1[v,\vec{a},b,i+1] =
		\max_{\vec{d}, b_1,b_2} \DP_1[v,\vec{d},b_1,i] + \DP[u_{i+1},\vec{a}-\vec{d},b_2] + 
		\w(v u_i) \cdot b_2.
	\end{equation}
	Here, we select $\vec{d}$ such that $d_j\in [a_j]_0$ for each $j\in [\max_{\ex}]$ and $b_1,b_2$ are selected to be in $[b]_0$.
	We finally set $\DP[v,\vec{a},b] := \DP_1[v,\vec{a},b,z]$.

	We return \yes, if $\DP[\rho,\vec{a}^*,1] \ge D$ where $\rho$ is the root of the given phylogenetic tree \Tree and $a_i^* = H_i$. Otherwise, we return \no.

\proofpara{Correctness}
	In the basic case, it is enough to check whether $a_i \ge \ell(x)$ for $i\ge \ex(x)$.
	A visualization is given in \Cref{fig:example-4b}.
	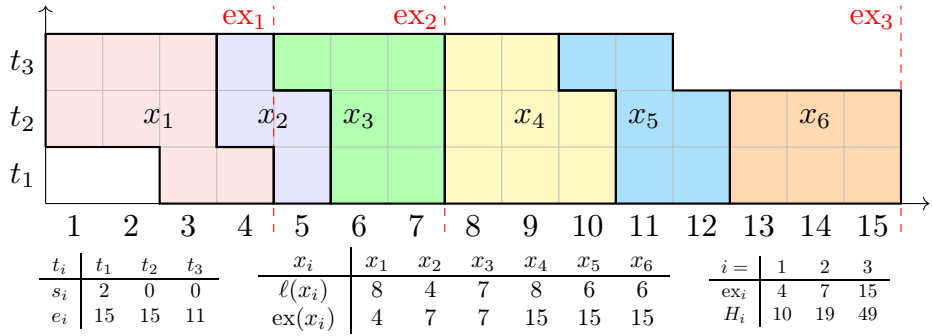
\begin{figure}[t]
		\centering
		\begin{tikzpicture}[scale=0.75,every node/.style={scale=1.2}]
			\tikzstyle{txt}=[circle,fill=none,draw=none,inner sep=0pt]
			
			\fill[red!10] (0,1) rectangle (3,3);
			\fill[red!10] (2,0) rectangle (4,1);
			\node[txt] at (2,1.5) {$x_1$};
			
			\fill[blue!10] (3,1) rectangle (4,3);
			\fill[blue!10] (4,0) rectangle (5,2);
			\node[txt] at (4,1.5) {$x_2$};
			
			\fill[green!30] (4,2) rectangle (5,3);
			\fill[green!30] (5,0) rectangle (7,3);
			\node[txt] at (5.5,1.5) {$x_3$};
			
			\fill[yellow!30] (7,0) rectangle (9,3);
			\fill[yellow!30] (9,0) rectangle (10,2);
			\node[txt] at (8.5,1.5) {$x_4$};
			
			\fill[cyan!30] (9,2) rectangle (11,3);
			\fill[cyan!30] (10,0) rectangle (12,2);
			\node[txt] at (10.5,1.5) {$x_5$};
			
			\fill[orange!30] (12,0) rectangle (15,2);
			\node[txt] at (13.5,1.5) {$x_6$};

			\draw[gray!50] (0,1) -- (15,1);
			\draw[gray!50] (0,2) -- (15,2);
			\draw[gray!50] (0,3) -- (11,3);
			
			\draw[gray!50] (1,1) -- (1,3);
			
			\foreach \i in {2,...,11}
				\draw[gray!50] (\i,0) -- (\i,3);
			
			\foreach \i in {12,...,15}
				\draw[gray!50] (\i,0) -- (\i,2);
				
			\foreach \i in {1,...,15}
				\node[txt] at (\i-.5,-.4) {\i};
				
			\foreach \i in {1,...,3}
				\node[txt] at (-.4,\i-.5) {$t_\i$};
			
			\draw[->] (0,0) -- (0,3.5);
			\draw[->] (0,0) -- (15.5,0);
			
			\foreach \i in {4,7,15}
				\draw[dashed,red] (\i,-.5) -- (\i,3.5);
			
			\node[txt,red] at (3.5,3.25) {$\ex_1$};
			\node[txt,red] at (6.5,3.25) {$\ex_2$};
			\node[txt,red] at (14.5,3.25) {$\ex_3$};
			
			\draw[thick] (0,1) -- (0,3) -- (3,3) -- (3,1) -- (4,1) -- (4,0) -- (2,0) -- (2,1) -- (0,1);
			\draw[thick] (3,3) -- (3,1) -- (4,1) -- (4,0) -- (5,0) -- (5,2) -- (4,2) -- (4,3) -- (3,3);
			\draw[thick] (5,0) -- (5,2) -- (4,2) -- (4,3) -- (7,3) -- (7,0) -- (5,0);
			\draw[thick] (7,3) -- (7,0) -- (10,0) -- (10,2) -- (9,2) -- (9,3) -- (7,3);
			\draw[thick] (10,0) -- (10,2) -- (9,2) -- (9,3) -- (11,3) -- (11,2) -- (12,2) -- (12,0) -- (10,0);
			\draw[thick] (12,2) -- (12,0) -- (15,0) -- (15,2) -- (12,2);
		\end{tikzpicture}
		\resizebox{.2\columnwidth}{!}{%
			\begin{tabular}{c|ccc}
				$t_i$ & $t_1$ & $t_2$ & $t_3$\\
				\hline
				$s_i$ & 2 & 0 & 0\\
				$e_i$ & 15 & 15 & 11
			\end{tabular}
		}
		\;\;
		\resizebox{.45\columnwidth}{!}{%
			\begin{tabular}{c|cccccc}
				$x_i$ & $x_1$ & $x_2$ & $x_3$ & $x_4$ & $x_5$ & $x_6$\\
				\hline
				$\ell(x_i)$ & 8 & 4 & 7 & 8 & 6 & 6\\
				$\ex(x_i)$ & 4 & 7 & 7 & 15 & 15 & 15
			\end{tabular}
		}
		\;\;
		\resizebox{.2\columnwidth}{!}{%
			\begin{tabular}{c|ccc}
				$i=$ & $1$ & $2$ & $3$\\
				\hline
				$\ex_i$ & 4 & 7 & 15\\
				$H_i$ & 10 & 19 & 49
			\end{tabular}
		}
		\caption{This is a hypothetical schedule of taxa $x_1,\dots,x_6$.
			This schedule is only possible, as $x_2$ is already started before $\ex_1$.
			Thus, it is not sufficient to assign $x_2,x_3\in Y_2$ the person-hours between $\ex_1$ and $\ex_2$, which are $H_2-H_1=9$.}
		\label{fig:example-4b}
	\end{figure}%

	This algorithm is similar to the algorithm in \Cref{prop:T+maxr}\ref{prop:s-T^maxr}.
	Instead of remembering the available teams each timeslot, in $\vec{a}$,
	we store the available person-hours per unique remaining time.
	Therefore the correctness proof is analogous to the correctness proof of \Cref{prop:T+maxr}\ref{prop:s-T^maxr}.

\proofpara{Running time}
	The table $\DP$ contains $\Oh(n \cdot (H_{\max_{\ex}})^{\var_{\ex}})$ entries. Each of them can be computed in $\Oh(\log(D))$ time.
	Also the table $\DP_1$ contains $\Oh(n \cdot (H_{\max_{\ex}})^{\var_{\ex}})$ entries.
	In Recurrence~(\ref{eqn:T*maxr}), there are $\Oh((H_{\max_{\ex}})^{\var_{\ex}})$ options to choose $\vec{d}$ and $4$ to choose $b_1$ and $b_2$,
	and so the value of each entry can be computed in $\Oh((H_{\max_{\ex}})^{\var_{\ex}} \cdot n \cdot \log(D))$ time.
	Altogether we can compute a solution for \tPDs in $\Oh((H_{\max_{\ex}})^{2\var_{\ex}} \cdot n \cdot \log(D))$~time.
\end{proofm}
\begin{proofm}{of \Cref{prop:T+maxr}\ref{prop:ws-T+maxr}}
	\proofpara{Table definition}
	Let $\vec{A} = (A_1,\dots,A_{\max_{\ex}})$ be a $\max_{\ex}$-tuple in which $A_i$ are subsets of $T$.
	
	We define a dynamic programming algorithm with a table $\DP$ in which in entry $\DP[v,\vec{A},b]$ with $b\in \{0,1\}$ we store  0 if $b=0$;
 	if $b = 1$, then we store the maximum value of $\PDgiveTree{\Tree_v}(\{S\})$ for a subset of taxa $S$ below $v$ that can be saved using only teams from $A_j$ 
	at each timeslot $j\in [\max_{\ex}]$.
	We define an auxiliary table $\DP_1$ in which in entry $\DP[v,\vec{a},b,i]$ we only consider the first $i$ children of $v$.
	
	\proofpara{Algorithm}
	For each leaf $x$ we store $\DP[x,\vec{A},b] = 0$ if $b=0$ or
	if there is a team $t_j\in T$ and an integer $i\in [\max_{\ex} - \ell(x)]_0$ such that $t_j \in A_{i+1} \cap A_{i+2} \cap \dots \cap A_{i + \ell(x)}$.
	Otherwise, we store $\DP[x,\vec{A},1] = -\infty$.
	
	Let $v$ be an internal vertex with children $u_1,\dots,u_z$.
	We define $\DP_1[v,\vec{A},b,1] := \DP[u_1,\vec{A},b] + \w(v u_1) \cdot b$
	and we compute further values with the recurrence
	\begin{equation}
		\label{eqn:T+maxr}
		\DP_1[v,\vec{A},b,i+1] =
		\max_{\vec{B}, b_1,b_2} \DP_1[v,\vec{B},b_1,i] + \DP[u_{i+1},\vec{A}-\vec{B},b_2] + 
		\w(v u_i) \cdot b_2.
	\end{equation}
	Here, we select $\vec{B}$ such that $B_j\subseteq A_j$ for each $j\in [\max_{\ex}]$ and $b_1,b_2$ are selected to be in $[b]_0$.
	We finally set $\DP[v,\vec{A},b] := \DP_1[v,\vec{A},b,z]$.

	We return \yes, if $\DP[\rho,\vec{A^*},1] \ge D$ where $\rho$ is the root of the given phylogenetic tree \Tree and $A_i^*$ for $i\in[\max_{\ex}]$ contains team $t_j =(s_j,e_j) \in T$ if and only if~$s_j<i\le e_j$. Otherwise, we return \no.

	\proofpara{Correctness}
	If $\DP$ stores the intended value, $\Instance$ is a \yes-instance if and only if $\DP[\rho,\vec{A^*},1] \ge D$, where $\rho$ is the root of the phylogenetic $X$-tree \Tree.
	It remains to prove that $\DP$ stores the intended value.
	
	For a leaf $x$, if $\DP[x,\vec{A},b]=0$ then either $b=0$ or there are $\ell(x)$ consecutive $A_i$s in which a team $t_j$ occurs. Thus, $x$ can be saved and $\PDgiveTree{\Tree_x}(\{x\}) = 0$.
	Likewise we show it the other way round and the basic case is correct.
	
	The correctness of the recurrence can be shown analogously to the correctness of \Cref{prop:T+maxr}\ref{prop:s-T^maxr}.

	\proofpara{Running time}
	Observe that for each $i\in [\max_{\ex}]$ there are $2^{|T|}$ options to select~$A_i$. Thus, there are $(2^{|T|})^{\max_{\ex}} = 2^{|T| \cdot \max_{\ex}}$ options for~$\vec{A}$.
	Therefore, the table $\DP$ contains $\Oh(n \cdot 2^{|T| \cdot \max_{\ex}})$ entries of which the internal vertices can be computed in $\Oh(\log(D))$ time.
	For a leaf $x$, we need to iterate over the teams and the timeslots to check whether there is a team contained in $\ell(x)$ consecutive $A_i$s. Thus, we can compute all entries of $\DP$ in $\Oh(2^{|T| \cdot \max_{\ex}} \cdot |T|^{1+\max_{\ex}} \cdot n)$.
	
	Also the table $\DP_1$ contains $\Oh(n \cdot 2^{|T| \cdot \max_{\ex}})$ entries.
	In Recurrence~(\ref{eqn:T+maxr}) observe that $A_j'$ is a subset of $A_j$, so there are $\Oh(3^{|T| \cdot \max_{\ex}})$ viable options for $\vec{A}$ and $\vec{B}$. Thus, all entries of $\DP_1$ are computed in $\Oh(3^{|T| \cdot \max_{\ex}} \cdot n \cdot \log(D))$ time.
	
	Subsequently, because $|T|^{1+\max_{\ex}} \le 1.5^{|T|\cdot\max_{\ex}} \cdot |T|$ we can compute a solution for an instance of \tPDws in $\Oh(3^{|T| \cdot \max_{\ex}} \cdot |T| \cdot n \cdot \log(D))$ time.
\end{proofm}

\subsection{Few time lengths and remaining times}
The scheduling problem $1||\sum w_j(1-U_j)$ is \FPT with respect to $\var_\ell + \var_{\ex}$~\cite{hermelin}.
In this subsection, we describe a dynamic programming algorithm (similar to the approach in \cite[Thm.~3]{GNAP}) to prove that \tPDs is at least \XP when parameterized by $\var_\ell + \var_{\ex}$.
Here, $\var_\ell$ and $\var_{\ex}$ are the number of unique needed time lengths and unique remaining times, respectively.

Let \Instance be an instance of \tPDs,
let $\ell(X) = \{\ell_1,\dots,\ell_{\var_\ell}\}$ for each $\ell_i < \ell_{i+1}$ with $i\in [\var_\ell-1]$ be the unique time lengths, and let $\ex(X) = \{\ex_1,\dots,\ex_{\var_{\ex}}\}$ with $\ex_i < \ex_{i+1}$ for each $i\in [\var_{\ex}-1]$ be the unique remaining times.

\begin{proposition}
	\label{prop:varl+varr}
	\tPDs can be solved in $\Oh(n^{2\var_\ell \cdot \var_{\ex}} \cdot (n + \var_\ell\cdot \var_{\ex}^2))$ time.
\end{proposition}
\begin{proof}
	\proofpara{Table definition}
	By~$\vec{A}$ we denote an integer-matrix of size $\var_\ell \times \var_{\ex}$.
	Further, we denote $A_{i,j}$ to be the entry in row $i\le \var_\ell$ and column $j\le \var_{\ex}$ of~$\vec{A}$.
	With~$\vec{A}_{(i,j) +z}$ we denote the matrix resulting from~$\vec{A}$ in which in row $i$ and column $j$, the value~$z$ is added.
	
	We define a dynamic programming algorithm with table $\DP$ in which entry $\DP[v,\vec A,b]$ stores $0$ if $b=0$ or the maximum diversity that can be achieved in the subtree rooted at $v$ in which at most $A_{i,j}$ taxa $x$ are chosen with $\ell(x) = \ell_i$ and $\ex(x) = \ex_j$.
	We define an auxiliary table $\DP_1$ in which in entry $\DP[v,\vec{A},b,i]$ we only consider the first $i$ children of $v$.
	
	\proofpara{Algorithm}
	For each leaf $x$ with $\ell(x) = \ell_i$ and $\ex(x) = \ex_j$ we store $\DP[x,\vec A,b] = 0$ if $b=0$ or $A_{i,j} > 0$.
	Otherwise, we store $\DP[x,\vec A,b] = -\infty$.
	
	Let $v$ be an internal vertex with children $u_1,\dots,u_z$.
	We define $\DP_1[v,\vec{A},b,1] := \DP[u_1,\vec{A},b] + \w(v u_1) \cdot b$
	and we compute further values with the recurrence
	\begin{equation}
		\label{eqn:varl+valr}
		\DP_1[v,\vec A,b,i+1] = \max_{\vec{B},b_1,b_2} \DP_1[v,\vec{B},b_1,i] + \DP[u_{i+1},\vec{A}-\vec{B},b_2] + \w(v u_{i+1})\cdot b_2.
	\end{equation}
	Here, we select $\vec{B}$ such that $B_{i,j}\le A_{i,j}$ for each $i\le [\var_\ell]$, $j\le [\var_{\ex}]$
	and $b_1,b_2$ are selected to be in $[b]_0$.
	We finally set $\DP[v,\vec{A},b] := \DP_1[v,\vec{A},b,z]$.
	
	We return \yes if $\DP[\rho,\vec{A}] \ge D$ for the root $\rho$ of \Tree and a matrix $\vec{A}$ such that $(\ell_1,\dots,\ell_{\var_\ell}) \cdot \vec A \cdot \vec 1_i \le H_i$ for each $i\in [\var_{\ex}]$.
	Here, $\vec 1_i$ is a $\var_{\ex}$-dimensional vector in which the first $i$ positions are 1 and the remaining are 0 for each $i\in [\var_i]$.
	
	\proofpara{Correctness}
	By the definition of $\DP$, we see that an instance is a \yes-instance of \tPDs if and only if $\DP[\rho,\vec{A}] \ge D$ for the root $\rho$ of \Tree and a matrix $\vec{A}$ such that $(\ell_1,\dots,\ell_{\var_\ell}) \cdot \vec A \cdot \vec 1_i \le H_i$ for each $i\in [\var_{\ex}]$.
	It only remains to show that $\DP$ stores the correct value.
	
	Observe that a leaf $x$ can be saved if and only if we are allowed to save a taxon with $\ell(x) = \ell_i$ and $\ex(x) = \ex_j$.
	Thus $\DP[x,\vec{A},b]$ should store 0 if $b=0$ or if $A_{i,j}\ge 1$.
	The other direction follows as well.
	
	The correctness of the recurrence can be shown analogously to the correctness of \Cref{prop:T+maxr}\ref{prop:s-T^maxr}.
	
	\proofpara{Running time}
	The matrix $\vec{A}$ contains $\var_\ell\cdot \var_{\ex}$ entries with integers in $[n]_0$.
	We can assume that if $A_{i,j} = n$ for some $i\le [\var_\ell]$, $j\le [\var_{\ex}]$,
	then $A_{p,q} = 0$ for each $p\ne i$ and $q\ne j$.
	Such that there are $n^{\var_{\ex}\cdot \var_{\ex}} + n$ options for a matrix $\vec{A}$.
	Therefore, the tables $\DP$ and $\DP_1$ contain $\Oh(n\cdot n^{\var_{\ex}\cdot \var_{\ex}})$ entries.
	
	Each entry in $\DP$ can be computed in linear time.
	To compute entries of~$\DP_1$, in Recurrence~(\ref{eqn:varl+valr}) we iterate over possible matrices $\vec{B}$ and $b_1,b_2$.
	Therefore we can compute each entry in $\DP_1$ in $\Oh(n^{2\var_{\ex}\cdot \var_{\ex}+1})$ time.
	
	To initialize we again iterate over possible matrices $\vec{A}$ and compute whether $(\ell_1,\dots,\ell_{\var_\ell}) \cdot \vec A \cdot \vec 1_i \le H_i$ in $\Oh(n^{\var_{\ex}\cdot \var_{\ex}} \cdot \var_\ell\cdot \var_{\ex}^2)$ time.
	That proves the overall running time.
\end{proof}

\subsection{Pseudo-polynomial running time on stars}
In this final subsection, we show that \tPDs can be solved in pseudo-polynomial time if the input tree is a star.
In the light of \Cref{prop:ws-Dbar}, such a result is unlikely to hold for \tPDws. 

\begin{proposition}
	\label{prop:stars}
	If the given input tree is a star, \tPDs can be solved
	\begin{itemize}
		\itemsep-.35em
		\item in $\Oh((\max_{\ex})^2 \cdot n \cdot \log(D))$ time,
		\item in $\Oh(D^2 \cdot n \cdot \log(H_{\max_{\ex}}))$ time,
		\item in $\Oh(\Dbar^2 \cdot n \cdot \log(H_{\max_{\ex}}))$ time, or
		\item in $\Oh((\max_\w)^2 \cdot n^3 \cdot \log(H_{\max_{\ex}}))$ time. $\max_\w$ is the maximum edge weight.
	\end{itemize}
\end{proposition}
\begin{proof}
	Let $\Instance$ be an instance of \tPDs in which the given phylogenetic tree is a star.
	With the help of \KP, we solve the problem separately on the different classes $Y_i$ and use a dynamic programming algorithm to combine the solutions.
	In \KP, we are given a set of items $N$, each with a weight $\w_i$ and a profit $p_i$, a limited capacity $C$, and a desired profit $P$. The question is whether we can select a subset of the items $S$ such that the sum of profits $p_i$ in $S$ is at least $P$ while the sum of weights in not exceeding $C$.
	
	A solution for an instance of \KP can be found in~$\Oh(C \cdot |N| \cdot \log(P))$ time~\cite{weingartner,rehs} via a dynamic programming algorithm, indexing solutions by capacity $C' \leq C$ and number of items $i \leq |N|$, and storing the maximum profit for a subset of the first $i$ items that has total weight $\leq w$.
	We note that such an algorithm also computes the maximum profit for every capacity $C'\leq C$, and so we may assume that in $\Oh(C \cdot |N| \cdot \log(P))$ time it is possible to construct a table $DP_a$ such that $DP_a[i,C']$ stores the maximum $P$ such that $(Y_i,\w_j',p_j,C,P)$ is a \yes-instance of \KP. 
	
	Along similar lines, there exists a dynamic programming algorithm with running time~$\Oh(P \cdot |N|^2 \cdot \log(C))$ time~\cite{weingartner,rehs}, where solutions are indexed by $P'\leq P$ and  $i \leq |N|$, and we store the minimum capacity $C'$ for which there is a subset of the first $i$ items with total profit $\geq P'$.
	Adapting that algorithm, we can also get a running time of~$\Oh(\Pbar \cdot |N|^2 \cdot \log(\Cbar))$, where $\overline P = \sum_{a_i\in N} p_i - P$ and $\Cbar = \sum_{a_i\in N} \w_i - C$. In such an algorithm we index solutions by $\Pbar' \leq \Pbar$ and $i \leq N$, and store the \emph{maximum} total weight of a subset of the first $i$ items whose total profit is \emph{at most} $\Pbar'$. If such a set exists with weight at least $\Cbar$ for $\Pbar' = \Pbar$, then the complement set is a solution for the \KP instance.

	\proofpara{Algorithms}
	We describe the algorithm with running time $\Oh((\max_{\ex})^2 \cdot n \cdot \log(D))$ and omit the similar cases for the other variants.
	
	In the following for each class $Y_i$ in \Instance, we consider an instance of \KP with items $Y_i$ in which $x_j\in Y_i$ with incoming edge $e_j$ has weight $\w_j':=\ell(x_j)$ and profit $p_j:=\w(e_j)$.
	Define a table $\DP_a$ in which for any $i \in [var_r]$ and $C \in [\ex_i]$,  $\DP_a[i,C]$ stores the maximum desired profit $P$ such that $(Y_i,\w_j',p_j,C,P)$ is a \yes-instance of \KP. 
	
	We next define a table $\DP_1$  in which we combine the sub-solutions on $Z_i$.
	As a base case, we store $\DP_1[1,C] := \DP_a[1,C]$ for each $C\in [\ex_1]$, $P\in [D]$, and $\Pbar\in [\Dbar]$.
	To compute further values, we can use the recurrence $\DP_1[i+1,C] = 
			\max_{C' \in [C]_0} \DP_1[i,C'] + \DP_a[i+1,C-C']$.
	Finally, we return \yes if $\DP_1[\var_{\ex},\max_{\ex}] \ge D$.
			

	\proofpara{Correctness}
	For convenience we will write $\w(S)$ to denote $\sum_{x\in S} \w(\rho x)$ for a set of taxa. Note that $\w(S) = \PD(S)$ as $\Tree$ is a star.
	The correctness of the values in $\DP_a$ follows from the correctness of the \KP-algorithms.
	Assume that for some $i\in [\var_{\ex}]$ the correct value is stored in $\DP_1[i,C]$ for each $C\in D$.
	
	Let $\DP_1[i+1,C]$ store $a\ge 0$. Then
	by construction 
	there is a $C'$ such that $a = a_1 + a_2 = \DP_1[i,C'] + \DP_a[i+1,C-C']$.
	Subsequently, there are sets $S_1\subseteq Z_i$ and $S_2\subseteq Y_{i+1}$ such that $\ell(S_1) + \ell(S_2) = C' + (C-C') = C$ and $\w(S_1) + \w(S_2) = a_1 + a_2 = a$.
	Therefore, $S := S_1 \cup S_2$ is a set with $\ell(S) = C$ and $\DP_1[i+1,C] = \w(S)$.
	
	To finish the correctness, let $S\subseteq Z_{i+1}$ be a set with $\ell(S) = C$.
	Define $S_1 := S\cap Z_{i}$ and $S_2 := S\cap Y_{i+1}$ and let $C'$ be $\ell(S_1)$.
	We conclude that~$\w(S) = \w(S_1) + \w(S_2) \ge \DP_1[i,C'] + \DP_a[i+1,C-C']$.
	
	\proofpara{Running time}
	Observe that the computation in \KP is also with a dynamic programming algorithm, such that all the values of the table $\DP_a$ are computed in $\Oh((\max_{\ex}) \cdot n \cdot \log(D))$ time.
	
	The table $\DP_1$ has $\var_{\ex}\cdot \max_{\ex}$ entries.
	To compute a value with the recurrence, we need to check the $\Oh(\max_{\ex})$ options to select $C$ and add two numbers of size at most $D$, and so we can compute all values of $\DP_1$ in $\Oh((\max_{\ex})^2 \cdot \var_{\ex} \cdot \log(D))$ time.
	
	For the other running times:
	Observe that the computation for \KP can be done in $\Oh(D^2 \cdot n \cdot \log(H_{\max_{\ex}}))$, and in $\Oh(\Dbar^2 \cdot n \cdot \log(H_{\max_{\ex}}))$ time.
	The rest follows analogously.
	
	We can assume that $D \le n \cdot \max_\w$, as we are dealing with a trivial \no-instance otherwise.
	Therefore we can solve \tPDs in $\Oh((n \cdot \max_\w)^2 \cdot n \cdot \log(H_{\max_{\ex}})) = \Oh((\max_\w)^2 \cdot n^3 \cdot \log(H_{\max_{\ex}}))$ time.
\end{proof}

\section{Conclusion}
\label{sec:discussion}
With \tPDs and \tPDws, we introduced two \NP-hard problems in which one can optimize phylogenetic diversity under the restriction that taxa have distinct extinction times.
We have proven that both problems are \FPT with respect to the target diversity $D$ and that \tPDs is also \FPT when parameterized by the acceptable loss of phylogentic diversity $\Dbar$.

As \tPDs is proven to be weakly \NP-hard, it remains an open question whether \tPDs is solvable in pseudo-polynomial time.
Indeed, we do not know if \tPDs or \tPDws can be solved in polynomial time even when the maximum time length needed to save a taxon is 2.
We note that the scheduling problem $1||\sum w_j (1-U_j)$ is \Wh 1-hard when parameterized by the unique number of processing times~\cite{heeger2024}; this implies that even on stars \tPDws is \Wh 1-hard when parameterized by the unique number of time lengths.

We ask whether the  $\Oh(2^{n})$ running time for \tPDs (\Cref{prop:X}) can be improved to $\Oh(2^{o(n)})$, or whether this bound can be shown to be tight under \SETH or \ETH.
It also remains an open question whether  \tPDs or \tPDws are \FPT with respect to $\max_{\ex}$.
We have not regarded kernelization algorithms. An interesting open question therefore is whether \tPDs or \tPDs have a kernelization of polynomial size with respect to $D$ or $\Dbar$.

\end{document}